\documentclass[acmtocl]{acmtrans2m}
\usepackage{amsmath}
\usepackage{amssymb}
\usepackage{epsfig}
\usepackage{eclbip}
\usepackage{epic}
\newenvironment{grp}{\begin{center}\begin{tabular}{l}}{\end{tabular}\end{center}}
\def\ebox{\quad{\vrule height6pt width6pt depth0pt}}

\newtheorem{theorem}{Theorem}[section]

\newtheorem{corollary}[theorem]{Corollary}
\newtheorem{proposition}[theorem]{Proposition}
\newtheorem{subproperty}[theorem]{Subproperty}
\newtheorem{invariant}[theorem]{Invariant}
\newtheorem{lemma}[theorem]{Lemma}

\newdef{definition}[theorem]{Definition}
\newdef{remark}[theorem]{Remark}
\newtheorem{example}{Example}[section]

\newcommand{\rarw}{\ensuremath{\rightarrow}}
\newcommand{\Rarw}{\ensuremath{\Rightarrow}}
\newcommand{\lrarw}{\ensuremath{\leftrightarrow}}
\newcommand{\lorarw}{\ensuremath{\longrightarrow}}
\newcommand{\larw}{\ensuremath{\leftarrow}}
\newcommand{\Larw}{\ensuremath{\Leftarrow}}
\newcommand{\mcl}[1]{\ensuremath{\mathcal{#1}}}

\newcommand{\dbiI}{\ensuremath{\mbox{\textbf{I}}}}
\newcommand{\dbiJ}{\ensuremath{\mbox{\textbf{J}}}}
\newcommand{\defn}{\textit}
\newcommand{\assign}{\ensuremath{\overset{\mbox{\tiny def}}{=}}}

\newcommand{\qrank}{\ensuremath{\mbox{qrank}}}

\newcommand{\ucv}{\ensuremath{\text{UCV}}}
\newcommand{\ufo}{\ensuremath{\text{UFO}}}
\newcommand{\fo}{\ensuremath{\mcl{FO}}}
\newcommand{\bld}[1]{\ensuremath{\mbox{\textbf{{#1}}}}}
\newcommand{\strG}{\ensuremath{\bld{G}}}
\newcommand{\gaif}{\ensuremath{\Bbb{G}}}
\newcommand{\strT}{\ensuremath{\bld{T}}}

\newcommand{\substruct}{\ensuremath{\subseteq}}

\newcommand{\db}{\ensuremath{\mathcal{D}}}
\newcommand{\code}{\ensuremath{\text{code}}}
\newcommand{\bit}[2]{\ensuremath{\text{bit}_{#1}(#2)}}
\newcommand{\ef}{Ehrenfeucht-Fra\"iss\'e}

\title{Logical Queries over Views: Decidability and
Expressiveness\footnote{A preliminary
version of this paper appeared in \cite{DBLP:conf/icdt/BaileyD99}}}

\author{JAMES BAILEY, \\ The University of Melbourne
    \and GUOZHU DONG, \\ Wright State University
    \and ANTHONY WIDJAJA TO \\ University of Edinburgh}

\begin{abstract}
We study the problem of deciding satisfiability of first order logic queries
over views, our aim being to delimit the boundary between the decidable and the
undecidable fragments of this language. Views currently occupy a central place
in database research, due to their role in applications such as information
integration and data warehousing. Our main result is the identification
of a decidable class of first order queries over unary conjunctive views that generalises 
the decidability of the classical class of first order sentences over unary 
relations, known as the L\"{o}wenheim class.
We then demonstrate how various extensions of this class lead to undecidability
and also provide some expressivity results. Besides its theoretical interest,
our new decidable class is potentially interesting for use in applications such
as deciding implication of complex dependencies,
analysis of a restricted class of active database rules, and ontology reasoning.
\end{abstract}

\category{F4.1}{MATHEMATICAL LOGIC AND FORMAL LANGUAGES}{Mathematical Logic}
\category{H2.3}{DATABASE MANAGEMENT}{Languages}

\terms{Theory}

\keywords{Satisfiability, containment, unary view, decidability,
first order logic,
database query, database view, conjunctive query,
L\"{o}wenheim class,  monadic logic, unary logic, ontology reasoning}

% ...

\begin{document}
\maketitle

\section{Introduction}

The study of views in relational databases has attracted much
attention over the years.  Views are an indispensable
component for activities such as
data integration
and data warehousing \cite{Wid95,GQPRS95,LRO96}, where they
can be used as
``mediators'' for source information that is not directly accessible to
users.
This is especially helpful in
modelling the integration of data from
 diverse sources, such as legacy systems and/or the world wide web.

Much of the research related to views has addressed fundamental problems
such as containment and rewriting/optimisation of queries
using views (e.g. see \cite{U97,Levy01}).
In this paper, we examine the use of views in a somewhat different context,
where they are used as the basic unit for writing logical expressions.
We provide results on the related decision problem in this paper, for a range
of possible view definitions.  In particular, for the case where
views are monadic/unary conjunctive queries, we show that the corresponding
query logic is decidable.  This corresponds to an interesting new fragment
of first order logic.
On the application side, this decidable query language also has some interesting potential applications
for areas such as implication of complex dependencies, ontology reasoning and termination results for active rules.

\subsection{Informal Statement of the Problem}
Consider a relational vocabulary
$R_1, \ldots, R_p$
and a set of views
$V_1, \ldots , V_n$.
Each view definition corresponds to a first order
formula over the vocabulary. Some example views (using horn clause style
notation) are

\begin{grp}
$V_1(x_1,y_1) \leftarrow R_1(x_1,y_1), R_2(y_1,y_1,z_1), R_3(z_1,z_2,x_1), R_4(z_2,x_1)$\\
$V_2(z_1) \leftarrow R_1(z_1,z_1)$
\end{grp}

Each such view can be expanded into to a first order sentence, e.g.
$V_1(x_1,y_1) \Leftrightarrow \exists z_1,z_2 (R_1(x_1,y_1) \wedge
R_2(y_1,y_1,z_1), R_3(z_1,z_2,x_1) \wedge \neg R_4(z_2,x_1))$. A
{\em first order view query} is a first order formula expressed {\em
solely} in terms of the given views. e.g. $q_1=\exists x_1,y_1
((V_1(x_1,y_1) \vee V_1(y_1,x_1)) \wedge \neg V_2(x_1)) \wedge
\forall z_1 (V_2(z_1) \Rightarrow V_1(z_1,z_1))$ is an example first
order view query, but $q_2=\exists x_1,y_1 (V_1(x_1,y_1) \vee
R(y_1,x_1))$ is not.  By expanding the view definitions, every first
order view query can clearly be re-written to eliminate the views.
Hence, first order view queries can be thought of as a fragment of
first order logic, with the exact nature of the fragment varying
according to how expressive the views are permitted to be.

From a database perspective, first order view queries are particularly suited to
applications where the source
data is unavailable, but summary data (in the form of views) is.
Since many database and reasoning
languages are based on first order logic (or extensions
thereof), this makes it a useful choice for manipulating the views.

Our purpose in this paper is to determine, for what types of view definitions,
satisfiability (over both finite and infinite models) is decidable for the language.
If views can be binary, then this language is clearly as
powerful as first order logic over binary base relations,
and hence undecidable (see \cite{BGG96}).
The situation becomes far more interesting, when we restrict the
form that views may take --- in particular, when their arity must be unary.
Such a restriction has the effect of constraining which parts of the
underlying database can be ``seen'' by the view formula and also constrains
how such parts may be connected.

\subsection{Contributions}

The main contribution of this paper is the definition of a
language called the {\em first order unary conjunctive
view language} ($\ucv$) and a proof of its decidability.
As its name suggests, it uses unary arity views
defined by conjunctive queries\footnote{More generally, views may be
any existential formulas with one free variable, since this can be rewritten
into a disjunction of conjunctive formulas with one free variable.}.
We demonstrate that it is a maximal decidable class, in
the sense that increasing the expressiveness of the view definitions
results in undecidability.
Some interesting aspects of this decidability result are:

\begin{itemize}
\item It is well known that first order logic solely over monadic relations
is decidable \cite{lowenheim}, but the extension to dyadic
relations is undecidable \cite{BGG97}.   The first order unary conjunctive
view language can be seen as an interesting intermediate case between the two,
since although only monadic predicates (views) appear in the query, they
are intimately related to database relations of higher arity.
\item The language is able to express some interesting
properties, which might be applied to various kinds of reasoning over
ontologies.
It can also be thought of as a powerful generalisation of unary
inclusion dependencies \cite{CKV90}.  Furthermore, it has an interesting
characterisation as a decidable class of rules (triggers) for active databases.
\end{itemize}

To briefly give a feel for this decidable language, we next provide
some example unary conjunctive
views and a first order unary conjunctive view query defined over them:
\begin{grp}
$V_1(x) \leftarrow R_1(x,y), R_2(y,z), R_3(z,x'), R_4(x',x)$\\
$V_2(x) \leftarrow R_1(x,y), R_1(x,z), R_4(y,z) $\\
$V_3(x) \leftarrow R_1(x,y), R_1(x,z), R_4(y,y), R_4(z,x) $\\
$V_4(x) \leftarrow R_1(x,y), R_3(y,z), R_4(z,x'), R_4(x',y'), R_3(y',x) $\\
$\exists x (V_2(x) \wedge \neg V_1(x)) \wedge \neg\exists y (V_3(y) \wedge \neg V_4(y))$
\end{grp}

\subsection{Paper Outline}

The paper is structured as follows:
Section \ref{sec:prelim} defines the necessary
preliminaries and background concepts.
Section \ref{sec:ucvdef} presents the definition of the logic $\ucv$.
Section \ref{sec:decid} is the core section of the paper, where
the decidability result for the class $\ucv$ is proved.
Section \ref{extend}
shows that extensions to the language, such as allowing
negation, inequality or recursion in
views,
result in undecidability.
Section \ref{sec:app}
covers applications of the decidability results and then Section \ref{sec:expressive}
provides some results on expressivity.
Section \ref{sec:related} discusses related work
and section \ref{sec:summary}
summarises and discusses future work.

\section{Preliminaries}
\label{sec:prelim}
In this section, we state basic definitions and relevant results. The
reader is assumed to be familiar with standard results and notations
from mathematical logic (e.g. see \cite{End2001}). In the following,
formulas are always first-order. The symbol $\fo$ denotes the set of first
order formulas over any vocabulary $\sigma$. In addition, if
$\mcl{L} \subseteq \fo$ (i.e. $\mcl{L}$ is a fragment of $\fo$), we denote by
$\mcl{L}(\sigma)$ the set of formulas in $\mcl{L}$ over the vocabulary
$\sigma$.

\subsection{First-order logic}
A (relational) \defn{vocabulary} $\sigma$ is a tuple $\langle R_1,\ldots,R_n
\rangle$ of relation symbols with each $R_i$ associated with a
specified arity $r_i$. A (relational) $\sigma$-\defn{structure} $\bld{A}$ is the
tuple
\[
    \langle A; R_1^{\bld{A}},\ldots,R_n^{\bld{A}} \rangle
\]
where $A$ is a non-empty set, called the \defn{universe}
(of $\bld{A}$), and
$R_i^{\bld{A}}$ is an $r_i$-ary relation over $A$ interpreting $R_i$. We refer
to the elements in the set $A$ as \defn{the elements in $\bld{A}$}, or simply
by \defn{constants}\footnote{Although it is common in mathematical logic to use
the term ``constants'' to mean the interpretation of constant symbols in the
structure, no confusion shall arise in this article, as we assume the absence
of constant symbols in the vocabulary. Our results, nevertheless, easily extend
to vocabularies with constant symbols.} (of $\bld{A}$). In the
sequel, we write
$R_i$ instead of $R_i^{\bld{A}}$ when the meaning is clear from the context. We
also use $STRUCT(\sigma)$ to denote the set of all $\sigma$-structures. We
assume a countably infinite set VAR of variables. An \defn{instantiation} (or
\defn{valuation}) of a structure $\bld{\dbiI}$ is a function $v: \text{VAR}
\rarw I$. Extend this function to free tuples (i.e. tuple of variables) in the
obvious way. We use the
usual Tarskian notion of satisfaction to define $\dbiI \models \phi[v]$, i.e.,
whether $\phi$ is true in $\dbiI$ under $v$. If $\phi$ is a sentence, we
simply write $\dbiI \models \phi$. The \defn{image} of a structure
$\dbiI$ under a formula $\phi(x_1,\ldots,x_n)$ is
\[
    \phi(\dbiI) \assign \{ v(x_1,\ldots, x_n): \text{$v$ is an instantiation of
            \dbiI, and $\dbiI \models \phi[v]$} \}.
\]
In particular, if $n = 0$, we have that $\phi(\dbiI) \neq \emptyset$
iff $\dbiI \models \phi$. We say that two
$\sigma$-structures $\bld{A}$ and $\bld{B}$ {\em agree}
on ${\cal L}$ iff for all $\phi \in {\cal L}(\sigma)$ we have
$\bld{A} \models \phi \Leftrightarrow \bld{B} \models \phi$.

Following the convention in database theory, the \defn{(tuple) database
$\db(\bld{A})$
corresponding to the structure $\bld{A}$} (defined above) is the set
\[
    \{ R_i(t) : \text{$1 \leq i \leq n$ and $t \in R_i^{\bld{A}}$} \}.
\]
It is easy to see that such a database can be considered a structure with
universe
$adom(\bld{A})$, which is defined to be the set of all elements of $\bld{A}$
occurring in at least one relation $R_i$, and relations built appropriately from
$\db(\bld{A})$. Abusing terminologies, we refer to the elements of
$\db(\bld{A})$ as \defn{tuples (associated with $\bld{A}$)}. In addition, when
the
meaning is clear from the context, we shall also abuse the term \defn{free
tuple} to mean an atomic formula $R(u)$, where $R \in \sigma$ and $u$ is a
tuple of variables.

A formula $\phi$ is said to be \defn{satisfiable} if there exists a structure
$\bld{A}$ (either of finite or infinite size) such that $\phi(\bld{A}) \neq
\emptyset$; such a structure is said to be a \defn{model} for $\phi$. We say
that $\phi$ is \defn{finitely satisfiable} if there exists a finite structure
$\dbiI$ such that $\phi(\dbiI) \neq \emptyset$. Without loss of generality, we
shall focus only on  sentences when we are dealing with the satisfiability
problem. In fact, if $\phi$ has some free variables, taking its existential
closure preserves satisfiability [Indeed we shall see that the languages we consider are
closed under first-order quantification].

Given two $\sigma$-structures $\bld{A}, \bld{B}$, recall that $\bld{A}$ is a
\defn{substructure} of $\bld{B}$ (written $\bld{A} \substruct \bld{B}$) if
$A \subseteq B$ and $R^{\bld{A}} \subseteq R^{\bld{B}}$ for every relation
symbol $R$ in $\sigma$. We say that $\bld{A}$
is an \defn{induced} substructure of $\bld{B}$ (i.e. \defn{induced} by $A
\subseteq B$) if for
every relation
symbol $R$ in $\sigma$, $R^{\bld{A}} = R^{\bld{B}} \cap A^r$, where $r$ is the arity of
$R$. Now, a \defn{homomorphism} from \bld{A} to \bld{B} is a function
$h: A \rarw B$ such that, for every relation symbol $R$ in $\sigma$ and
$\bld{a} = (a_1,\ldots,a_r) \in R^{\bld{A}}$, it is the case that
$h(\bld{a}) \assign (h(a_1),\ldots,h(a_r)) \in R^{\bld{B}}$. An
\defn{isomorphism} is a  bijective homomorphism whose inverse is a homomorphism.

The \defn{quantifier rank} $\qrank(\phi)$ of of a formula $\phi$
is the maximum nesting depth of quantifiers in $\phi$.


\subsection{Views}
For our purpose, a \defn{view} over $\sigma$ can be thought of as an
arbitrary FO formula over $\sigma$. We say that a view $V$ is
\defn{conjunctive} if it can be written as a \defn{conjunctive query}, i.e. of
the form
\[
    \exists x_1,\ldots,x_n( R_1(u_1) \wedge \ldots \wedge R_k(u_k) )
\]
where each $R_i$ is a relation symbol, and each $u_i$ is a
\defn{free tuple} of appropriate arity. We adopt the horn clause
style notation for writing conjunctive views. For example, if
$\{y_1,\ldots,y_n\}$ is the set of free variables in the above
conjunctive query, then we can rewrite it as
\[
    V(y_1,\ldots,y_n) \larw R_1(u_1), \ldots, R_k(u_k)
\]
where $V(y_1, \ldots, y_n)$ is called the \defn{head} of $V$, and
the conjunction $R_1(u_1), \ldots, R_k(u_k)$ the \defn{body} of $V$.
The \defn{length} of the conjunctive view $V$ is defined to be the sum of the
arities of the relation symbols in the \emph{multiset} $\{R_1,\ldots,R_k\}$.
For example, the lengths of the two views $V$ and $V'$ defined as
\begin{eqnarray*}
    V(x) & \larw & E(x,y) \\
    V'(x) & \larw & E(x,y), E(y,z)
\end{eqnarray*}
are, respectively, two and four. Additionally, if $n = 1$ (i.e. has a head of
arity 1), the view is said to be unary. {\em Unless
stated otherwise, we shall say ``view'' to mean ``unary-conjunctive view
with neither equality nor negation in its body''}.

\subsection{Graphs}
We use standard definitions from graph theory (e.g. see \cite{Dies05}).
A \defn{graph} is a structure $\strG = (G,E)$ where $E$ is a binary relation.
The \defn{girth} of a graph is the length of its shortest cycle.
For two vertices $x,y \in G$, we denote their distance by $d_{\strG}(x,y)$ (or
just $d(x,y)$ when $\strG$ is clear from the context). For two sets $S_1$ and
$S_2$ of vertices in $\strG$, we define their distance to be
\[
    d_{\strG}(S_1,S_2) := \min\{ d_{\strG}(a,b) : \text{$a \in S_1$ and
                            $b \in S_2$} \}.
\]
In a weighted graph $\strG$ with weight $w_{\strG}: E \rarw \mathbb{N}$, the
weight $w_{\strG}(P)$ of a path $P$ in $\strG$ is just $\sum_{e \in E(P)}
w_{\strG}(e)$. We shall write $w$ instead of $w_{\strG}$ if the meaning is
clear from the context. In the sequel, we shall frequently mention trees and
forests. We always assume that any tree has a selected node, which we call a
\defn{root} of the tree. Given a tree $\strT = (T,E)$, we can partition $T$
according to the distance of the vertices from the root.

The \defn{Gaifman graph} (see \cite{Gai82}) associated with a structure
$\bld{A}$ is the weighted undirected multi-graph $\gaif(\bld{A}) =
(G,E)$ such that:
\begin{enumerate}
    \item $G = A$.
    \item The multi-set $E$ is defined as follows: for each $x,y \in G$, we
    put an $R(t)$-labeled edge $xy$ in $E$ with weight $r$ (the arity
    of $R$) iff $x$ and $y$ appear in a tuple $R(t)$ in $\db(\bld{A})$.
    [Notice that the multiplicity of $xy$ in $E$
    depends on the number of tuples in $\db(\bld{A})$ that contain both $x$
    and $y$ as their arguments.]
\end{enumerate}
Note also that the subgraph of $\gaif(\bld{A})$ induced by the set of all
elements of $\bld{A}$ in a tuple $t$ is the complete graph $K_r$, and so an
$L$-labelled edge is
adjacent to an edge $e \in E$ iff all $L$-labelled edges are adjacent (i.e.
connected) to the
edge $e$. For any $a,b \in A$, we define the \defn{distance} $d_{\bld{A}}(a,b)$
between $a$ and $b$ to be their distance in $\gaif(\bld{A})$. Also, extend this
distance function to tuples and sets of tuples by interpreting them as sets of
elements of $\bld{A}$ that appear in them. Any pair
of tuples $R(t)$ and $R'(t')$ in $\db(\bld{A})$ are said to be \defn{connected}
(in $\bld{A}$) if in $\gaif(\bld{A})$ some (and hence all) $R(t)$-labeled edge
is adjacent to some (and hence all) $R'(t')$-labeled edge.

\subsection{Unary formulas}
A \defn{unary formula} is an arbitrary FO formula \emph{without equality} such
that
each of its relation symbols has arity one. Let $\sigma$ be a vocabulary whose
relation symbols are of arity one. We shall use $\ufo(\sigma)$ to denote the
set of all unary formulas without equality over $\sigma$. Also, we define
$\ufo = \cup_{\sigma} \ufo(\sigma)$.  The following lemma will be useful for
proving expressiveness results in Section \ref{sec:expressive}.

\begin{lemma}
For every unary sentence, there exists an equivalent one of quantifier rank 1.
\label{lm:unaryform}
\end{lemma}
\begin{proof}
By a straightforward manipulation. See the proof of lemma 21.12 in
\cite{BBJ2002}. [Their proof actually gives more than the result they claim. In
fact, their construction converts an arbitrary unary sentence into one with one
unary variable and of quantifier rank 1.]
\end{proof}

\subsection{Ehrenfeucht-Fra\"isse Games}
We shall need a limited form of Ehrenfeucht-Fra\"isse games; for a general
account, the reader may consult \cite{Libkin2004}. The games
are played by two players, Spoiler and Duplicator, on two $\sigma$-structures
$\bld{A}$ and $\bld{B}$. The goal of Spoiler is to show that the structures
are different, while Duplicator aims to show that they are the same. The game
consists of a single round. Spoiler chooses a structure (say, $\bld{A}$) and an
element $a$ in it, after which Duplicator has
to respond by choosing an element $b$ in the other structure $\bld{B}$.
Duplicator wins the game iff the substructure of $\bld{A}$ induced
by $\{a\}$ is isomorphic to the substructure of $\bld{B}$ induced by $\{b\}$.
Duplicator has a winning strategy iff Duplicator has a winning move, regardless
of how Spoiler behaves.
\begin{proposition}[(Ehrenfeucht-Fra\"isse Games)]
Duplicator has a winning strategy on $\bld{A}$ and $\bld{B}$ iff $\bld{A}$
and $\bld{B}$ agree on first-order formulas over $\sigma$ of quantifier rank 1.
\end{proposition}

\subsection{Other Notation}
Regarding other notation we shall use
throughout the rest of the paper: we shall use $a,b$ for constants,
$x,y,z$ for variables, $u$ for free tuples, $U,V$ for views, $\mcl{U},
\mcl{V}$ for sets of views, $\sigma$ for vocabularies, $R_1, R_2,\ldots$ for
relation symbols,
$\bld{A}, \bld{B},\ldots$ for structures and
$A,B$ for their respective  universes. If $\db$ is a database (a set of
tuples), we use $adom(\db)$ to denote the set of constants in
$\db$. Finally,
given a $a \in adom(\db)$ and a ``new'' constant $b \notin \db$, we define
$\db[b/a]$ to be the database that is obtained from $\db$ by replacing every
occurrence of $a$ by $b$. The notation $\db[b_1/a_1,\ldots,b_n/a_n]$
is defined in the same way.

\section{Definition of First Order Unary-conjunctive-view Logic}
\label{sec:ucvdef}
Let $\sigma$ be an arbitrary vocabulary, and $\mcl{V}$ be a finite set of (unary conjunctive)
views over $\sigma$, which we refer to as a \defn{$\sigma$-view set}. We now
inductively define the set
$\ucv(\sigma,\mcl{V})$ of \defn{first order unary-conjunctive-view (UCV) queries/formulas
over the vocabulary $\sigma$ and a $\sigma$-view set $\mcl{V}$}:
\begin{enumerate}
\item if $V \in \mcl{V}$, then $V(x) \in \ucv(\sigma,\mcl{V})$; and
\item if $\phi, \psi \in \ucv(\sigma,\mcl{V})$, then the formulas
    $\neg \phi, \phi \wedge \psi$ and $\exists x\phi$ belong to
    $\ucv(\sigma,\mcl{V})$.
\end{enumerate}
The smallest set of so-constructed formulas defines the set
$\ucv(\sigma,\mcl{V})$. We denote the set of all UCV formulas over the
vocabulary $\sigma$ by $\ucv(\sigma)$, i.e. $\ucv(\sigma) \assign
\bigcup_{\mcl{V}} \ucv(\sigma,\mcl{V})$ where $\mcl{V}$ may be any
$\sigma$-view set. Further, the set of all UCV queries is denoted
by $\ucv$, i.e. $\ucv \assign \bigcup_{\sigma} \ucv(\sigma)$, where $\sigma$
is any vocabulary. As usual, we use the shorthands
$\phi \vee \psi$, $\phi \rarw \psi$, $\phi \lrarw \psi$, and $\forall x\phi$
for (respectively) $\neg (\neg \phi \wedge \neg \psi), \neg \phi \vee \psi,
(\phi \rarw \psi) \wedge (\psi \rarw \phi)$, and $\neg \exists x \neg \phi$.
Thus, the UCV language is closed under boolean combinations and first-order
quantifications. As an example, consider the UCV formula
\[
    q_1 = \exists x( V(x) \wedge \neg V'(x) )
\]
where $V$ and $V'$ are defined as
\begin{eqnarray*}
    V(x) & \larw & E(x,y) \\
    V'(x) & \larw & E(x,y), E(y,z) \\
\end{eqnarray*}
This formula asserts that there exists a
vertex from which there is an outgoing arc, but no outgoing directed walk
of length 2.

Let us make a few remarks on the expressive power of the logic $\ucv$ with respect
to other logics. It is easy to see that the UCV language strictly subsumes \ufo $\;$ (the L\"owenheim class
without equality \cite{lowenheim,BGG97}), as UCV queries can be defined  over
 \emph{any} relational vocabularies
(i.e. including ones that include $k$-ary relation symbols with $k >
1$). It is also easy to see that allowing any general existential
positive formula (i.e. of the form
$\exists\overline{x}\phi(\overline{x})$ where $\phi$ is a
quantifier-free formula with no negation) with one free variable,
does not increase the expressive power of the logic. Indeed, the
quantifier-free subformula $\phi$ can be rewritten in disjunctive
normal form without introducing negation, after which we may
distribute the existential quantifier across the disjunctions and
consequently transform entire formula to a disjunction of
conjunctive queries with one or zero free variables. Each such
conjunctive query can then be treated as a view.

%

There are two ways in which we can interpret a UCV formula. The standard way is to
think of a UCV query as an FO formula over the underlying vocabulary. Take
the afore-mentioned query $q_2$ as an example. We can interpret this
query as the formula
\[
    \exists x( \exists y,z(E(x,y) \wedge E(y,z)) \wedge \neg \exists y(E(x,y) )
\]
over the graph vocabulary. The non-standard way is to regard a UCV query $\phi$
as a unary formula over the view set. For example, we can think of
$q_2$ as a unary formula over the vocabulary $\sigma' = \langle V,V' \rangle$.
Now, if $\phi \in \ucv(\sigma,\mcl{V})$, then we denote by $\phi^{\mcl{V}}$ the
unary formula over $\mcl{V}$ corresponding to $\phi$ in the non-standard
interpretation of UCV queries.
However, for notational convenience, we shall write $\phi$ instead of
$\phi^{\mcl{V}}$ when the meaning is clear from the context. Given a vocabulary
$\sigma$ and a $\sigma$-view set $\mcl{V} = \{V_1,\ldots,V_n\}$, we may define
the function $\Lambda: STRUCT(\sigma) \rarw STRUCT(\mcl{V})$ such that for any
$\dbiI \in STRUCT(\sigma)$
\[
    \Lambda(\dbiI) \assign \langle I; V_1^{\Lambda(\dbiI)}, \ldots,
                    V_n^{\Lambda(\dbiI)} \rangle
\]
where $V_i^{\Lambda(\dbiI)} \assign V_i(\dbiI)$.
For example, let $\sigma = \langle E \rangle$ and $\mcl{V} = \{V,V'\}$ be as
above, and let
\[
    \dbiI = \langle \{1,2,3,4\}; E^{\dbiI} = \{(1,2),(2,3),(3,4)\} \rangle.
\]
Then, we have
\[
    \dbiJ \assign \Lambda(\dbiI) = \langle \{1,2,3,4\}, V^{\dbiJ} = \{1,2,3\},
                V'^{\dbiJ} = \{1,2\} \rangle.
\]
In the following, we shall reserve the symbol $\Lambda$ to denote this special
function. In addition, if $\dbiJ \in STRUCT(\mcl{V})$ and
there exists a structure $\dbiI \in STRUCT(\sigma)$ such that $\Lambda(\dbiI) =
\dbiJ$, we say that the structure $\dbiJ$ is \defn{realizable} with respect
to the vocabulary $\sigma$ and the view set $\mcl{V}$, or that $\dbiI$
\defn{realizes} $\dbiJ$. We shall omit mention of $\sigma$ and $\mcl{V}$ if they
are understood by context.

A number of remarks about the notion of realizability are in order. First,
some unary structures are \emph{not} realizable with respect to a given
view
set $\mcl{V}$. For example, the query $q_2$
has infinitely many models if treated as a unary formula, but none of these
models are realizable, since $V' \subseteq V$. Second, if
$\phi \in \ucv(\sigma,\mcl{V})$ has a model $\dbiI$, then the
structure $\Lambda(\dbiI)$ over $\mcl{V}$ is a model for $\phi^{\mcl{V}}$.
In other words, if a UCV query is satisfiable, then it is
also satisfiable if treated as a unary formula. Conversely, it is also clearly
true that a UCV query is satisfiable, if it is satisfiable when treated as a
unary formula and that at least one of its models is realizable. More precisely,
if $\Lambda(\dbiI)$ is a model for $\phi^{\mcl{V}}$, then $\dbiI$ is a model
for $\phi$. So, combining these, we have
$I \models \phi$ iff $\Lambda(I) \models \phi^{\mcl{V}}$. So, we immediately have the following lemma:
\begin{lemma}
    Suppose $\bld{A}, \bld{B} \in STRUCT(\sigma)$ and
    $\phi \in \ucv(\sigma,\mcl{V})$. Then, for
    $\Lambda: STRUCT(\sigma) \rarw STRUCT(\mcl{V})$ defined above,
    the following statements are equivalent:
    \begin{enumerate}
    \item $\bld{A} \models \phi$ iff $\bld{B} \models \phi$,
    \item $\Lambda(\bld{A}) \models \phi^{\Lambda}$ iff $\Lambda(\bld{B}) \models
        \phi^{\Lambda}$.
    \end{enumerate}
\label{lem:realize}
\end{lemma}
This lemma is useful when combined with Ehrenfeucht-Fra\"isse games. For
example, suppose that we are given a model $\bld{A}$ for $\phi$, and we
construct a ``nicer'' structure $\bld{B}$ that, we wish, satisfies $\phi$. If
we can prove that the second statement in the lemma (which is often easier to
establish as views have arity one), we might deduce that $\bld{B} \models \phi$.

\section{Decidability of UCV Queries}
\label{sec:decid}
In this section, we prove our main result that satisfiability is decidable for $\ucv$ formulas.
Our main theorem stipulates that $\ucv$ has the bounded model property.

\begin{theorem}
\label{th:fin_s}
Let $\phi$ be a formula in $\ucv$. Suppose, further, that $\phi$ contains precisely the
views in the view set $\mcl{V}$, and relation symbols in the vocabulary
$\sigma$, with $m$ being the maximum length of the views in $\mcl{V}$,
and $p = |\sigma|$. If $\phi$ is satisfiable, then it has a model using at most
$2^{2^{q(p,m)}}$ elements, for some fixed polynomial $q$ in $p$ and $m$.
\end{theorem}
Before we prove this theorem, we first derive some corollaries. Simple
algebraic manipulations yield the following corollary.
\begin{corollary}
Continuing from Theorem \ref{th:fin_s}, if $n$ is the size of (the parse tree of) a
satisfiable formula $\phi$, then $\phi$ has a model of size
$2^{2^{g(n)}}$ for some fixed polynomial $g$ in $n$.
\label{th:fin_s2}
\end{corollary}
Corollary \ref{th:fin_s2} immediately leads to the decidability of 
satisfiability for $\ucv$. We can in fact derive a tighter bound.
\begin{theorem}
Satisfiability for the $\ucv$ class of formulas is in 2-NEXPTIME.
\label{th:timecomp}
\end{theorem}
This theorem follows immediately from the following proposition and corollary
\ref{th:fin_s2}.
\begin{proposition}
Let $s$ be a non-decreasing function with $s(n) \geq n$. Then, the problem of
determining whether an FO sentence has a model of size at most $s(n)$, where
$n$ is the size of the input formula, can be decided
nondeterministically in $2^{O(n \log(s(n)))}$ steps.
\label{prop:timecomp}
\end{proposition}
\begin{proof}
We may use any reasonable encoding $\code(\bld{A})$ of a finite structure
$\bld{A}$ in bits (e.g. see \cite[Chapter 6]{Libkin2004}). The size of the
encoding, denoted $|\bld{A}|$, is polynomial in $|A|$.
We first guess a structure $\bld{A}$ of size at most $s(n)$. Let $s' = |A|$.
Since the size $|\bld{A}|$ of the encoding of $\bld{A}$ is polynomial in $s'$,
the guessing procedure takes $O(s^k(n))$ time steps for some constant $k$. We,
then, use the usual procedure for evaluating whether $\bld{A} \models \phi$.
This can be done in $O(n \times |\bld{A}|^n)$ steps (e.g. see
\cite[Proposition 6.6]{Libkin2004}). Simple algebraic manipulations give
the sought after upper bound.
\end{proof}

Observe that a lower bound for satisfiability of $\ucv$ formulas follows immediately
from the NEXPTIME completeness for satisfiability of \ufo $\;$ formulas given in \cite{BGG97}

\begin{theorem}
Satisfiability for the $\ucv$ class of formulas is NEXPTIME hard.
\end{theorem}

What remains now is to prove theorem \ref{th:fin_s}.
\begin{proof}[of theorem \ref{th:fin_s}]
Let $\phi, m, p$ be as stated in theorem \ref{th:fin_s}. We begin by
first enumerating all possible views over $\sigma$ of length at most
$m$.   As we shall see later in the proof of Subproperty
\ref{subp:prune}, doing so will help facilitate the correctness of
our construction of a finite model, since enumerating all such views
effectively allows us to determine all possible ways the model may
be ``seen'' by views, or parts of views. Let $\mcl{U} =
\{V_1,\ldots,V_N\}$ be the set of all non-equivalent views obtained.
By elementary counting, one may easily verify that $N \leq m(mp)^m$.
Indeed, each view is composed of its head and its body, whose length
is bounded by $m$. The body is a set of conjuncts that we may fix in
some order. There are at most $m$ variables that the head can take.
Each position in the body is a variable ($m$ choices) that is part
of a relation $R$ ($p$ choices). The upper bound is then immediate.

Let $\dbiI_0$ be a (possibly infinite) model for $\phi$. [If it is infinite,
by the L\"owenheim-Skolem theorem, we may assume that it is countable.] Without loss
of generality, we may assume that there exists a ``universe'' relation $U$ in
$\dbiI_0$ which contains each constant in $adom(\dbiI_0)$.
Otherwise, if $U' \notin \sigma$ is a unary relation symbol, the
$(\sigma \cup \{U'\})$-structure obtained
by adding to $\dbiI_0$ the relation $U'$, which is to be interpreted as $I_0$,
is also a model for $\phi$.

Let us now define $2^N$ formulas $C_0,\ldots,C_{2^{N}-1}$ of the form
\[
    C_i(x) \assign (\neg) V_1(x) \wedge \ldots \wedge (\neg) V_N(x),
\]
where the conjunct $V_j(x)$ is negated iff the $j$th bit of the binary
representation of $i$ is 0. For each $\bld{A} \in STRUCT(\sigma)$, these
formulas induce an equivalence relation on $A$ with each set $C_i(\bld{A})$
being an equivalence class. When $\bld{A}$ is clear, we refer to the
equivalence class $C_i(\bld{A})$ simply as $C_i$.
In addition, the existence of the universe relation $U$ in $\dbiI_0$
implies that the all-negative equivalence class $C_0$ is empty.

We next describe a sequence of five satisfaction-preserving procedures for
deriving a finite model from $\dbiI_0$. This sequence is best described diagrammatically:
\[
    \dbiI_0 \stackrel{\text{\tt makeJF}}{\lorarw}
    \dbiI_1 \stackrel{\text{\tt rename1}}{\lorarw}
    \dbiI_2 \stackrel{\text{\tt rename2}}{\lorarw}
    \dbiI_3 \stackrel{\text{\tt copy}}{\lorarw}
    \dbiI_4 \stackrel{\text{\tt prune}}{\lorarw}
    \dbiI_5.
\]
The $i$th procedure above takes a structure $\dbiI_i$ as input, and outputs
another structure $\dbiI_{i+1}$. The structure $\dbiI_5$ is guaranteed to be
finite (and indeed bounded). That each procedure preserves satisfiability
immediately follows by subproperties \ref{subp:makeJHS}, \ref{subp:rename},
\ref{subp:maxim}, \ref{subp:copy}, and \ref{subp:prune}.
While reading the description of the procedures below, it is instructive to keep
in mind that the property that $C_i(\dbiI_j) = \emptyset$ iff
$C_i(\dbiI_{j+1}) = \emptyset$ is sufficient for showing that the $j$th
procedure preserves satisfiability (see lemma \ref{lm:equiv}).

Roughly speaking, the procedure \texttt{makeJF} transforms the initially given structure
$\dbiI_0$ into another structure that has a forest-like graphical
representation, called a ``justification forest''. Each
subsequent procedure works only on justification forests. In the sequel,
we shall use $\mcl{H}_i$ to denote our graphical representation of
$\dbiI_i$ ($i \in \{1,\ldots,5\}$).

\begin{flushleft}
    \itshape The procedure \texttt{makeJF}
\end{flushleft}
We define the structure $\dbiI_1$ by first defining a sequence
$\dbiI_1^0, \dbiI_1^1, \ldots$ of structures such that $\dbiI_1^k$
is a substructure of $\dbiI_1^{k+1}$, and then setting $\dbiI_1 =
\bigcup_{k=0}^{\infty} \dbiI_1^k$. [Note: we take the normal union,
not \emph{disjoint union}.] We first deal with the base case of
$\dbiI_1^0$. For each non-empty equivalence class $C_i(\dbiI_0)$, we
choose a witnessing constant $a_i \in C_i(\dbiI_0)$. We define
$I_1^0$ as the collection of all such $a_i$s. All relations in
$\dbiI_1^0$ are empty. Each $a_i$ is said to be \defn{unjustified}
in $\dbiI_1^0$, meaning that the model is missing tuples that can
witness the truth of $a_i$ being a member of some equivalence class.
We now describe how to define $\dbiI_1^{k+1}$ from $\dbiI_1^{k}$.
For each $a \in I_1^k$, if $a \in C_i(\dbiI_0)$ for some $i$, it is
the case that $a \in V_j(\dbiI_0)$ iff $\bit{j}{i} = 1$ for $1 \leq
j \leq N$. For such $a$, we may take a minimal witnessing
substructure $\bld{S}_a$ of $\dbiI_0$ such that $a \in
V_j(\bld{S}_a)$ iff $\bit{j}{i} = 1$. As each constant in
$adom(\bld{S}_a)$ appears in at least one relation in $\bld{S}_a$,
we shall often think of these witnessing structures as databases
(i.e. sets of tuples), and refer to them as \defn{justification
sets}. We define the structure $\dbiI_1^{k+1}$ to be the union of
$\dbiI_1^k$ and all the witnessing structures $\bld{S}_a$ such that
$a$ is unjustified in $\dbiI_1^k$. The elements in $I_1^k$ become
\emph{justified} in $\dbiI_1^{k+1}$. The elements in
$I_1^{k+1}-I_1^k$ are then said to be
\defn{unjustified}
in $\dbiI_1^{k+1}$.
Observe that the structure $\dbiI_1^{K+1}$ does not unjustify any elements that
were justified in $\dbiI_1^{k}$, since there is no negation in the view
definitions.
Finally, the structure $\dbiI_1$ is defined as the union
of all $\dbiI_1^k$s. Observe that each element in $I_1$ appears in at least one
relation in $\dbiI_1$.

The structure $\dbiI_1$ has an intuitive graphical representation, which we
denote by $\mcl{H}_1$. The graph $\mcl{H}_1$ is simply a labeled forest in
which each tree $T_i$ (for some $0 \leq i \leq 2^N-1$) corresponds to exactly one
witnessing constant $a_i$ for each non-empty $C_i$. We define $T_i$
as follows: the root of $T_i$ is labeled by
$\bld{S}_{a_i} \times C_i$; and for each $j=0,1,\ldots$, any $\bld{S}_b \times
C_k$-labeled node $v$ at
level $j$ (for some justification set $\bld{S}_b$ and equivalence class formula
$C_k$), and any constant $c$ in $adom(\bld{S}_b)$ that is distinct from $b$,
define
a new $\bld{S}_c \times C_{k'}$-labeled node to be a child of $v$, for the
unique $k'$ such that $c \in C_{k'}(\dbiI_0)$. In the following, when the meaning
is clear, we shall often refer to an $(\bld{S}_a \times C_k)$-labeled node
simply as a $\bld{S}_a$-labeled node. Also, observe
the similarity of the construction of $\mcl{H}_1$ and that of $\dbiI_1$. In
fact, the union of all $\bld{S}_a$, for which there is an $\bld{S}_a$-labeled
node in
$\mcl{H}_1$, is precisely $\dbiI_1$. Observe also that each tree $T_i$ may be
infinite. For obvious reasons, we shall refer to $T_i$ as a \defn{justification
tree} (of $a_i$), and to $\mcl{H}_1$ as \defn{justification forest}. In the
following, for any justification tree $T$ and any justification forest
$\mcl{H}$, their \defn{corresponding structures} (or \defn{databases}), denoted
by $\mcl{D}(T)$ and $\mcl{D}(\mcl{H})$ respectively, are defined to be the
union of all $\bld{S}_a$, such that there is an $\bld{S}_a$-labeled node in,
respectively, $T$ and $\mcl{H}$. Furthermore, we shall use
$adom(T)$ and $adom(\mcl{H})$ to denote $adom(\mcl{D}(T))$ and
$adom(\mcl{D}(\mcl{H}))$, respectively. The elements in the set $adom(T)$ and
$adom(T)$ and $adom(\mcl{H})$ are referred to as, respectively,
constants in $T$ and constants in $\mcl{H}$.


We now illustrate this procedure by a small example. Define the
UCV formula
\[
    \phi = \forall x( V_1(x) \wedge \neg V_2(x) ),
\]
where the views are
\begin{eqnarray*}
    V_1(x) & \larw & E(x,y) \\
    V_2(x) & \larw & E(x,x).
\end{eqnarray*}
Here, we have $\mcl{V} = \{V_1,V_2\}$, $\sigma = \langle E \rangle$, and
$m = 2$. Suppose that
\[
    \dbiI_0 = \langle \mathbb{N}, E = \{(0,1),(1,2),(2,3),(3,4),\ldots\} \rangle
\]
is a path extending indefinitely to the right. Then, we have
$\dbiI_0 \models \phi$. Enumerating all non-equivalent views over $\sigma$
of length at most $m$, we have $\mcl{U} = \{V_1,V_2,V_3\}$ where
\[
    V_3(x) \larw E(y,x).
\]
Now, there are exactly two non-empty equivalence classes:
\begin{eqnarray*}
    C_{100} & = & \{0\} \\
    C_{101} & = & \{1,2,\ldots\}.
\end{eqnarray*}
Then, we have $\bld{S}_0 = \{E(0,1)\}$ and $\bld{S}_i = \{E(i-1,i),E(i,i+1)\}$
for $i > 0$. Following the above procedure, we obtain the trees $T_{100}$
and $T_{101}$ as depicted in figure \ref{fig:H1}. Note that $\mcl{H}_1$ is
the disjoint union of $T_{100}$ and $T_{101}$.
\begin{figure}[h]
\begin{center}
    \epsfig{file=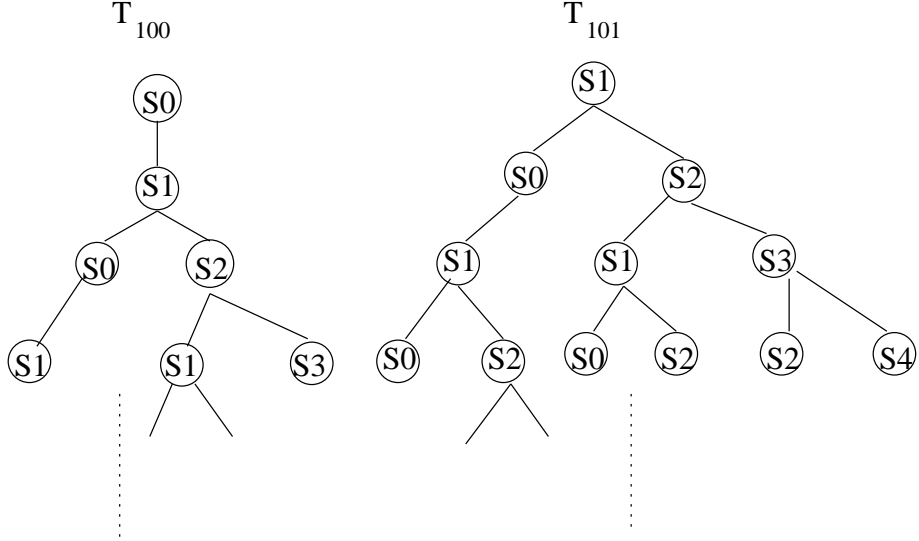}
\end{center}
\caption{\small
A depiction of the justification forest $\mcl{H}_1$ as an output of
\texttt{makeJF}.
\label{fig:H1}}
\end{figure}

\begin{flushleft}
    \itshape The procedure \texttt{rename1}
\end{flushleft}
\emph{Proviso: in subsequent procedures (including the present one), we shall
not change the second entries (i.e. $C_i$) of each node label (i.e. of the form
$\bld{S}_a\times C_i$) and frequently omit mention of them.}

The aim of this procedure is to ensure that there are no two justification trees
$T$ and $T'$ with $adom(T) \cap adom(T') \neq \emptyset$. It essentially
performs renaming of constants in $adom(T)$, for each tree
$T$ in $\mcl{H}_1$.
This step will later help us
guarantee the
correctness of the last step that is used to produce the final model $\dbiI_5$, which relies
on a kind of ``tree disjointness'' property.
More formally, we define $\dbiI_2$ to be the disjoint union\footnote{The
disjoint union of two $\sigma$-structures $\bld{A}$ and $\bld{B}$ with
$A \cap B = \emptyset$ is the structure with universe $A \cup B$ and
relation $R$ interpreted 
as $R^{\bld{A}}\cup R^{\bld{B}}$. If $A \cap B \neq \emptyset$, one
can simply force disjointness by renaming constants.} of
$\mcl{D}(T)$ over all trees $T$ in $\mcl{H}_1$. The justification forest
$\mcl{H}_2$ corresponding to $\dbiI_2$ can be obtained from $\mcl{H}_1$ by
renaming constants of the tuples in each tree $T$ in $\mcl{H}_1$ accordingly.

Let us continue with our previous example of $\mcl{H}_1$. The graph $\mcl{H}_2$
in this case will be precisely identical to $\mcl{H}_1$, except that in
$T_{101}$ we use the label, say, $\bld{S}_{0'} = \{E(0',1')\}$ (resp.
$\bld{S}_{i'} =
\{E((i-1)',i'),E(i',(i+1)')\}$ for $i > 0$) instead of $\bld{S}_0$ (resp.
$\bld{S}_i$ for $i > 0$).

\begin{flushleft}
    \itshape The procedure \texttt{rename2}
\end{flushleft}
The aim of this procedure is to transform the model in such a way
that each constant $a$ can appear only
at two consecutive levels, say $j$ and $j+1$, within each tree. It appears at level $j$  as part
of an $S_b$-labeled node $v$, for some constant $b \neq a$, and at level $j+1$
as part of an $S_a$-labeled node that is a child of $v$. Further, the procedure
ensures that any given constant occurs in at most one node's label at each
level in a tree.
Again, this will step will later help us
guarantee the
correctness of the step that is used to produce the final model $\dbiI_5$, which relies
on the existence of a kind of internal ``disjointness'' property within trees.

Let us fix a sibling ordering for the nodes
within each tree $T_i$ in $\mcl{H}_2$. Define a set
$U$ of constants disjoint from $I_2$ as follows:
\[
    U = \{ a_{j,l} : \text{$j, l \in \mathbb{N}$ and $a \in I_2$} \}.
\]
For $a,b \in I_2$, we require that $a_{j,l} \neq b_{j',l'}$ whenever
either $j \neq j'$, or $l \neq l'$, or $a \neq b$. For each tree $T_i$ and for each $j = 1,2,\ldots$,
choose the $l$th node $v$ with respect to the fixed sibling ordering (say,
$\bld{S}_a$-labeled)
at level $j$ in $T_i$.   Let $v$'s children be $v_1,\ldots,v_k$ (labeled by,
respectively, $\bld{S}_{b^1},\ldots,\bld{S}_{b^k}$ with $b^h \neq a$).
Now do the
following:
change $v$ to $\bld{S}_a[b^1_{j,l},\ldots,b^k_{j,l}/b^1,\ldots,b^k]$;
and change $v_h$, where $1 \leq h \leq k$, to
$\bld{S}_{b^h_{j,l}} \assign \bld{S}_{b^h}[b^h_{j,l}/b^h]$.
Observe that
there are two stages in this procedure where each non-root node at level $j$,
say $\bld{S}_a$-labeled, undergoes relabeling: first when we are at level $j-1$
(the constant $a$ is renamed by $a_{j,k}$ for some $k$), and second when we are
at level $j$ (constants other than $a_{j,k}$ are renamed
for what is now $\bld{S}_{a_{j,k}}$ ). The output of this
procedure on $\mcl{H}_2$ is denoted by $\mcl{H}_3$, whose corresponding
structure we denote by $\dbiI_3$.

Continuing with our previous example.  The root node  $u_1$ of $T_{100}$ in
$\mcl{H}_2$ is $\bld{S}_0=\{E(0,1)\}$, its child $u_2$ (sibling zero at level 1)
is $\bld{S}_1=\{E(0,1),E(1,2)\}$
and in turn the children of that child are $u_3=\bld{S}_0=\{E(0,1)\}$ (sibling 0 at level 2)
 and $u_4=\bld{S}_2=\{E(1,2),E(2,3)\}$ (sibling 1 at level 2).
Under the \texttt{rename2} procedure, node $u_1$
is unchanged, since it is at level zero.  Node $u_2$ is changed to $\bld{S}_1=\{E(0_{1,0},1),E(1,2_{1,0})\}$
Node $u_3$ is changed to $\bld{S}_{0_{1,0}}=\{E(0_{1,0},1_{2,0})\}$ and $u_4$ is changed to
$\bld{S}_{2_{1,0}}=\{E(1_{2,1},2_{1,0}),E(2_{1,0},3_{2,1})\}$.


\begin{flushleft}
    \itshape The procedure \texttt{copy}
\end{flushleft}
This procedure makes a number of isomorphic copies of the model $\mcl{H}_3$ and then unions them
together.
Duplicating the model in this way facilitates the construction of a bounded model by the
\texttt{prune} procedure, that will be described shortly.
Let
$\delta$ be the total number of constants that appear in some tuples from a
node label at level $h := cm$ in $\mcl{H}_3$, for some fixed $c \in \mathbb{N}$,
independent from $\phi$, whose value will later become clear in the proofs that follow. By virtue of
procedure \texttt{makeJF}, we are guaranteed that
each node in $\mcl{H}_3$ can have at most $N \times m$ children, where
$N \times m$ represents an upper bound on the number of constants each
justification set might contain. Since there are at most $2^N$ trees in
$\mcl{H}_3$, by elementary counting, we see that $\delta \leq 2^N \times
(N \times m)^h$. Now, letting $g := cm$, make $\Delta := \delta^{g}$
(isomorphic) copies of
$\mcl{H}_3$, each with a disjoint set of constants. That is, the node labeling of
each new copy of $\mcl{H}_3$ is isomorphic to that of $\mcl{H}_3$, except that is uses
disjoint set of constants. Let us call them the copies $\mcl{B}_1,
\ldots, \mcl{B}_{\Delta}$ (the original copy of $\mcl{H}_3$ is included).
So, we have $B_i \cap B_j = \emptyset$, for $i \neq j$. For
each tree $T_i$ in $\mcl{H}_3$, we denote by $T_i^k$ the isomorphic copy of
$T_i$ in $\mcl{B}_k$. Now, let
\[
    \mcl{H}_4 = \mcl{B}_1 \cup \ldots \cup \mcl{B}_{\Delta}.
\]
The structure corresponding to $\mcl{H}_4$ is denoted by $\dbiI_4$. In the
sequel, each node at level $h$ in $\mcl{B}_k$ is said to be a \defn{(potential) leaf} of
$\mcl{B}_k$.

\begin{flushleft}
    \itshape The procedure \texttt{prune}
\end{flushleft}
The purpose of this procedure is to transform $\mcl{H}_4$ into a finite model.
Intuitively, this is achieved by ``pruning'' all trees at level $h$ and then rejustifying
the resulting unjustified constants by ``linking'' them to a justification being used in some other
part of the model.  This is the most complex step in the entire sequences of procedures, and care
will be needed later to prove to ensure that satisfiability is not violated when constants are being rejustified.

We begin first by describing the connections that we wish to construct between
the different parts of the model.   Roughly speaking, the model we intend to construct somewhat
resembles
a $\delta$-regular graph, whose nodes are the copies $\mcl{B}_1 \cup \ldots \cup \mcl{B}_{\Delta}$
made earlier,
and where edges between copies indicate that
one copy is being used to make a new justification for a node at level $h$ in
another copy.

Firstly though, we state a proposition from extremal graph theory
(see \cite[Theorem 1.4' Chapter III]{Bol04}] for proof)
that can be used
to guarantee the existence of the kind of $\delta$-regular graph we intend to construct.

\begin{proposition}
Fix two positive integers $\delta, g$ and take an integer $\Delta$ with
\[
    \Delta \geq \frac{(\delta - 1)^{g-1} - 1}{\delta - 2}.
\]
Then, there exists a $\delta$-regular graph of size $\Delta$ with girth at least
$g$.
\end{proposition}
Using $\delta,g$ and $\Delta$ as defined in the \texttt{copy} procedure, this
proposition implies that there exists a $\delta$-regular graph $\bld{G}$ with
vertices $\{\mcl{B}_1, \ldots,\mcl{B}_{\Delta}\}$ and with girth at least $g$.
Let us now treat $\bld{G}$ as a directed graph, where each edge in $\bld{G}$
is regarded as two bidirectional arcs.

Observe that, for each vertex $\mcl{B}_k$, there is a bijection $out_k$ from
the set of leafs (nodes at height $h$) of $\mcl{B}_k$ to the set of arcs going out from $\mcl{B}_k$
in $\bld{G}$. We next take each leaf of $\mcl{B}_k$ in
turn. For a leaf $v$ (say, $\bld{S}_b$-labeled),
suppose that $out_k(v) = (\mcl{B}_k,\mcl{B}_{k'})$. Choose $i$ such that
$b \in C_i(\dbiI_4)$.
If the root of $T_i^{k'}$ is $\bld{S}_c$-labeled, for some $c \in I_4$, then we
delete all descendants of $v$ in $T_i^k$ and change $v$ to $\bld{S}_c[b/c]$.
In this way, we ``prune'' each of the trees in $\mcl{H}_4$, and link
each leaf node to the root node of another tree for the purpose of
justification. We denote by $\mcl{H}_5$ the resulting collection of interlinked models,
whose corresponding structure is denoted by $\dbiI_5$.  $\mcl{H}_5$ can
be thought of as a collection of interlinked forests, where each forest corresponds
to one of the copies $\{\mcl{B}_1, \ldots,\mcl{B}_{\Delta}\}$ and each forest is a collection
of trees.

Observe now that each ``tree'' in $\mcl{H}_5$ is of height $h$. Since there are
at most $\Delta \times 2^N$ trees in $\mcl{H}_5$, each of which has at most
$(N \times m)^{h+1}$ constants, we see that
\begin{eqnarray*}
    I_5 & \leq & (\Delta \times 2^N) \times (N \times m)^{h+1} \\
        & \leq & ((2^N \times (Nm)^{cm})^{cm} \times 2^N) \times
        (N\times m)^{cm+1}
\end{eqnarray*}
It is easy to calculate now that $I_5 \leq 2^{2^{q(p,m)}}$ for some polynomial
$q$ in $p$ and $m$.  We have thus managed to construct a bounded model $\dbiI_5$
which satisfies the original $\ucv$ formula $\phi$.
\end{proof}

We now prove the correctness of our construction for theorem \ref{th:fin_s}.
The proof is divided into a series of subproperties that assert the correctness
of each procedure in our construction. First, we prove a simple lemma.
\begin{lemma}
Let $\mcl{V}$ be a set of (unary) views over a vocabulary $\sigma$. Suppose
$\dbiI, \dbiJ$ are $\sigma$-structures such that $C(\dbiI)$ is non-empty iff
$C(\dbiJ)$ is
non-empty for each equivalence class formula $C$ constructed with respect to
$\mcl{V}$. Then, $\dbiI$ and $\dbiJ$ agree on $\ucv(\sigma,\mcl{V})$.
\label{lm:equiv}
\end{lemma}
\begin{proof}
By standard Ehrenfeucht-Fra\"isse argument, we see that $\Lambda(\dbiI)$ and
$\Lambda(\dbiJ)$ agree on
$\ufo(\mcl{V})$. Then, by lemma \ref{lem:realize}, we have that $\dbiI$ and
$\dbiJ$ agree on $\ucv(\sigma,\mcl{V})$.
\end{proof}

\begin{subproperty}[(Correctness of \texttt{MakeJF})]
\begin{enumerate}
\item For each node $v$ (say, $(\bld{S}_a \times C_i)$-labeled) of $\mcl{H}_1$,
$a \in C_i(\dbiI_1)$ with witnessing structure $\bld{S}_a$.
\item If $\dbiI_0 \models \phi$, then $\dbiI_1 \models \phi$.
\end{enumerate}
\label{subp:makeJHS}
\end{subproperty}
\begin{proof}
First, note that $\dbiI_1 \subseteq \dbiI_0$. Since conjunctive queries
are monotonic, we have $V(\dbiI_1) \subseteq V(\dbiI_0)$ for each view
$V \in \mcl{U}$. So, we have that $a \in V(\dbiI_1)$ implies that $a \in
V(\dbiI_0)$. In addition, for each constant $a \in I_1$, if $a \in V(\dbiI_0)$,
then $a \in V(\dbiI_1)$, which is witnessed at some $\bld{S}_a$-labeled node.
In turn, this implies that for $a \in I_1$, it is the case that
$a \in C(\dbiI_1)$ iff $a \in C(\dbiI_0)$. This proves the first statement.
Also, by construction, if $C_i(\dbiI_0)$ is non-empty, where $i \in
\{0,\ldots,2^N-1\}$, we know that one of its members belongs to $\dbiI_1$,
witnessed at the root of $T_i$. Therefore, we also have that
$C(\dbiI_0)$ is non-empty iff $C(\dbiI_1)$ is non-empty. In view of lemma
\ref{lm:equiv}, we conclude the second statement.
\end{proof}
At this stage, it is worth noting that, once $\mcl{H}_1$ has been constructed, the
subsequent procedures might modify the label $(\bld{S}_a\times C_i)$ ---
its name (e.g. from $\bld{S}_a\times C_i$ to $\bld{S}_{a'} \times C_i$ for some
new constant $a'$) as well as its contents (e.g. replacing each occurrence of a
tuple $R(a,a)$ by $R(a',a')$ for some new constant $a'$). Despite this, we wish to highlight
that one invariant is preserved by each of these procedures that have been described:
\begin{invariant}[(Justification Set)]
Suppose $\mcl{H}$ is a justification forest of a structure $\bld{I}$, and $v$ a
$(\bld{S}_a\times C_i)$-labeled node of $\mcl{H}$. Then, we have
$a \in C_i(\bld{I})$ with witnessing structure $\bld{S}_a$.
\label{inv:JS}
\end{invariant}
Subproperty \ref{subp:makeJHS} shows that this is satisfied by $\mcl{H}_1$.
In fact, that this invariant is preserved by the later procedures will be
almost immediate from the proof of correctness of the procedure. Hence, we
leave it to the reader to verify.

\begin{subproperty}[(Correctness of \texttt{rename1})]
If $\dbiI_1 \models \phi$, then $\dbiI_2 \models \phi$.
\label{subp:rename}
\end{subproperty}
\begin{proof}
In this procedure, we perform constant renaming for each tree $T_i$ in
$\mcl{H}_1$. For the purpose of this proof, let us denote the tree so obtained
by $T_i'$. Such a renaming induces a bijection $f_i: adom(T_i)\rarw adom(T_i')$.
Extend $f_i$ to tuples, structures, and trees in the obvious way.
Observe that the structures corresponding to the trees $f_i(T_i)$ and $T_i$ are
isomorphic. Now, in view of lemma \ref{lem:rename}, it is easy to
check that for each tree $T_i$ in $\mcl{H}_1$ and a constant $a$ in the
structure corresponding
to $T_i$, $a \in C_j(\dbiI_1)$ iff $f_i(a) \in C_j(\dbiI_2)$. By virtue of by
lemma \ref{lm:equiv}, we conclude our proof.
\end{proof}

\begin{lemma}
Suppose that $a$ is a constant in the structure corresponding to $T_i$ of
$\mcl{H}_1$. Then, $a \in V(\dbiI_1)$ iff $f_i(a) \in V(\dbiI_2)$.
\label{lem:rename}
\end{lemma}
\begin{proof}
($\Rarw$) By subproperty \ref{subp:makeJHS}, it is the case that $a \in
V(\bld{S}_a)$.
Since $f_i$ is a bijection, it is also true that $f_i(a) \in V(f_i(\bld{S}_a))$.

($\Larw$) Let $\bld{M}$ be a minimal set of tuples in $\dbiI_2$ such that
$f_i(a) \in
V(\bld{M})$. Observe that there is a one-to-one function mapping the set of
conjuncts in $V$ to $\bld{M}$. For each tree $T_j$ in $\mcl{H}_2$, let
$\bld{M}_j$ denote the members of $\bld{M}$ that can be found in
$T_j$. Note that $adom(\bld{M}_j) \cap adom(\bld{M}_{j'}) = \emptyset$ for $j
\neq j'$. Now, let $\bld{M}' = \bigcup_j f_j^{-1}(\bld{M}_j)$. It is not
hard to see that $a \in V(\bld{M}')$. Since $\bld{M}' \subseteq \dbiI_1$, we
have $a \in V(\dbiI_1)$.
\end{proof}

\begin{subproperty}[(Correctness of \texttt{rename2})]
If $\dbiI_2 \models \phi$, then $\dbiI_3 \models \phi$.
\label{subp:maxim}
\end{subproperty}
\begin{proof}
Define the function $\eta:I_3 \rarw I_2$ such
that $\eta(a_{j,k}) = a$. Note that $\eta$ is onto. Extend $\eta$ to tuples, and
sets of tuples in the obvious way.
In view of lemma \ref{lm:equiv}, it is sufficient to show that, for each
$a \in I_3$ and $i \in
\{0,\ldots,2^N-1\}$, $a \in C_i(\dbiI_3)$ iff $\eta(a) \in C_i(\dbiI_2)$.
In turn, it is enough to show that, $a \in V(\dbiI_3)$ iff $\eta(a) \in
V(\dbiI_2)$.

($\Rarw$) Take a minimal set $\bld{M}$ of tuples in $\dbiI_3$ such that
$a \in V(\bld{M})$. Then, we have $\eta(a) \in V(\eta(\bld{M}))$. Since
$\eta(\bld{M}) \subseteq \dbiI_2$,
we have $\eta(a) \in V(\dbiI_2)$.

($\Larw$) Since invariant \ref{inv:JS} holds for $\mcl{H}_2$, the fact that
$\eta(a) \in V(\dbiI_2)$ is witnessed
by $S_{\eta(a)} \in \mcl{H}_2$. Since $S_a$ and $S_{\eta(a)}$ are isomorphic
justification sets, we have that $a \in V(\dbiI_3)$ is justified by
$S_a \in \mcl{H}_3$.
\end{proof}

\begin{subproperty}[(Correctness of \texttt{copy})]
\begin{enumerate}
\item For each node $v$ (say, $(\bld{S}_a \times C_i)$-labeled) of $\mcl{H}_4$,
$a \in C_i(\dbiI_4)$ with witnessing structure $\bld{S}_a$.
\item If $\dbiI_3 \models \phi$, then $\dbiI_4 \models \phi$.
\end{enumerate}
\label{subp:copy}
\end{subproperty}
\begin{proof}


Similar to the proof of subproperty \ref{subp:rename}.
\end{proof}


\begin{subproperty}[(Correctness of \texttt{prune})]
If $\dbiI_4 \models \phi$, then $\dbiI_5 \models \phi$.
\label{subp:prune}
\end{subproperty}
\begin{proof}
Recall that there are $N_l \assign \delta \times \Delta$ leafs in $\mcl{H}_4$.
Let us order these nodes as $v_1,\ldots,v_{N_l}$. Suppose also that $v_i$
is labeled by $\bld{S}_{b_i}$ for some $b_i \in I_4$. By virtue of
\texttt{rename2}, we see that $b_i \neq b_j$ whenever $i \neq j$. Next,
we may think of the procedure \texttt{prune} as consisting of $N_l$ steps,
where at step $i$, the node $v_i$ has all its descendants removed (pruned)
and $v_i$ is changed to
$\bld{S}_{b_i}' \assign \bld{S}_{c_i}[b_i/c_i]$ for some $c_i \in I_4$. Letting
$\mcl{K}_0 \assign \mcl{H}_4$, we denote by $\mcl{K}_i$ ($i = 1,\ldots,N_l$)
the resulting model after executing $i$ steps on $\mcl{K}_0$.
The structure corresponding to $\mcl{K}_i$ is denoted by $\dbiJ_i$.

We wish to prove by induction on $0 \leq i < N_l$ that
\begin{description}
\item[(I)] For each $a \in J_{i+1}$ and $V \in \mcl{U}$,
$a \in V(\dbiJ_{i+1})$ iff $a \in V(\dbiJ_i)$.
\item[(II)] Invariant \ref{inv:JS} holds for $\dbiJ_{i+1}$.
\item[(III)] For each $a \in J_{i+1}$, we have $a \in C_i(\dbiJ_{i+1})$ iff
        $a \in C_i(\dbiJ_i)$.
\end{description}

Note that $J_{i+1} \subseteq J_i$. So, by lemma \ref{lm:equiv} and the fact
that invariant \ref{inv:JS} holds for the initial case $\bld{J}_0$ (from
proofs of previous subproperties), statement (III) will imply what we wish
to prove. It is easy to see that statement (III) is a direct consequence of
statement (I). It is also easy to show that statement (I) implies statement
(II). This follows since
firstly, at step $i+1$, we replace the content
of $\bld{S}_{b_i}$ by that of $\bld{S}_{c_i}$, except for substituting $b_i$
for $c_i$. Second, the elements $b_i$ and $c_i$ belong to the same equivalence
class in $\dbiJ_i$, and invariant \ref{inv:JS} holds for $\dbiJ_i$ by
induction. Therefore, it remains only to prove statement (I).

Let us now fix $i < N_l$, $a \in J_{i+1}$, and $V \in \mcl{U}$. It is simple
to prove that $a \in V(\dbiJ_i)$ implies $a \in V(\dbiJ_{i+1})$. This is
witnessed by tuples in the $\bld{S}_a$-labeled (or $\bld{S}_{b_i}'$-labeled if
$a = b_i$) node in $\mcl{K}_{i+1}$, which exists by construction.

Conversely, we take a minimal set $\bld{M}$ of tuples in $\dbiJ_{i+1}$ with
$a \in V(\dbiJ_{i+1})$, witnessed by the valuation $\nu$. Our aim is to find a
set $\bld{M}'$ of tuples in
$\dbiJ_i$ with $a \in V(\bld{M}')$. Let $\bld{M}_{b_i} \assign \bld{M} -
\db(\dbiJ_i)$. Intuitively, $\bld{M}_{b_i}$ contains the set of \textit{new}
tuples. These are tuples which did not exist in the structure $\dbiJ_i$ and have been
created specifically to justify the node whose descendants (justifications)
have just been pruned.
By construction, we have $\bld{M}_{b_i} \subseteq
\bld{S}_{b_i}'$, which implies that $adom(\bld{M}_{b_i}) \subseteq
adom(\bld{S}_{b_i}')$. Observe also that $b_i \in adom(t)$ for each tuple $t$
in $\bld{M}_{b_i}$; otherwise, $t$ would be a tuple in
$\bld{S}_{c_i} \subseteq \dbiJ_i$ (i.e. it would not be a new tuple). Define
\[
    \bld{L} := \{ t \in \bld{M} - \bld{M}_{b_i} : \text{$t$ is connected to
            some $t' \in \bld{M}_{b_i}$ in $\bld{M}$} \}.
\]
$\bld{L}$ consists of tuples that are connected to new tuples.
Also, let $\bld{L}' \assign \bld{M} - \bld{M}_{b_i} - \bld{L}$, i.e., the set
of all tuples of $\bld{M}$ that are \emph{not} connected to any (new) tuples in
$\bld{M}_{b_i}$. Note that $\bld{L} \cup \bld{L}' \subseteq \dbiJ_i$, and
that the sets $\bld{M}_{b_i}$, $\bld{L}$, and $\bld{L}'$ form a partition on
$\bld{M}$. Also, by definition, we have $adom(\bld{L}') \cap adom(\bld{M}_{b_i}
\cup \bld{L}) = \emptyset$. In the following, we define $\bld{M}_{c_i} \assign
\bld{M}_{b_i}[c_i/b_i]$. Note that $\bld{M}_{c_i} \subseteq \bld{S}_{c_i}
\subseteq \dbiJ_i$.

Before we proceed further, it is helpful to see how we partition $\bld{M}$
on a simple example. Suppose that the view $V$ is defined as
\[
    V(x_0) \larw E(x_0,x_1),E(x_1,x_2),R(x_3,x_4),R(x_4,x_5).
\]
Furthermore, suppose that we take the valuation $\nu$ defined as $\nu(x_i) = i$.
In this case, $\bld{M}$ can be described diagrammatically as follows
\[
    V(0) \larw E(0,1),E(1,2),R(3,4),R(4,5).
\]
Assume now that the only tuple in $\bld{M}$ that doesn't belong
to $\db(\dbiJ_i)$ is $E(0,1)$. Then, we have $\bld{M}_{b_i} = \{E(0,1)\}$. It
is easy to show that $\bld{L} = \{E(1,2)\}$ and $\bld{L}' =
\{R(3,4),R(4,5)\}$.

We next state a result regarding $\bld{L}$ that will shortly be needed.  It clarifies
the nature of a partition that exists for $\bld{L}$ and the relationships which hold
between the elements of the partition.

\begin{proposition}
We can find tuple-sets $\bld{A}, \bld{B} \subseteq \bld{L}$ such that:
\begin{enumerate}
    \item $\bld{A} \cap \bld{B} = \emptyset$,
    \item $\bld{A} \cup \bld{B} = \bld{L}$,
    \item $adom(\bld{A}) \cap adom(\bld{B}) = \emptyset$,
    \item $b_i \notin adom(\bld{B})$, and
    \item $adom(\bld{M}_{b_i}) \cap adom(\bld{A}) \subseteq \{b_i\}$.
\end{enumerate}
\label{prop:part}
\end{proposition}
The proof of this proposition can be found at the end of this section. We now
shall construct $\bld{M}' \subseteq \db(\dbiJ_i)$ such that $a \in V(\bld{M}')$. First,
we put $\bld{L}'$ in $\bld{M}'$. This does not affect our choice of
tuple-sets that replace $\bld{M}_{b_i}$, $\bld{A}$, and $\bld{B}$ as
$adom(\bld{L}') \cap adom(\bld{M}_{b_i} \cup \bld{L}) = \emptyset$
(i.e. the set of free tuples instantiated by $\bld{L}'$ and the set of free
tuples instantiated by $\bld{M}_{b_i} \cup \bld{L}$ share no common variables),
as we have noted earlier. There are two cases to consider:
\begin{description}
\item[case 1] $a = b_i$.
        Let $F$ be the set of all free tuples in the body of $V$ such
that $\{ \nu(u) : u \in F \} =
            \bld{M}_{b_i} \cup \bld{B}$.
        Suppose $X$ is the set of all variables in $V$.
        Let $\{y_1,\ldots,y_r\} \subseteq X$ be the set of variables in
        $F$ such that $\nu(y_j) = b_i$. With $y$ as a new variable, let
        $F' := F[y/y_1,\ldots,y_r]$. Define the new view $V'(y)$
        whose conjuncts are exactly $F'$:
        \[
        V'(y) \larw \bigwedge F'
        \]
        Trivially, we have $b_i \in V'(\bld{M}_{b_i} \cup \bld{B})$. Then,
        as
        $b_i \notin adom(\bld{B})$ by proposition \ref{prop:part}, we have
        $c_i \in V'(\bld{M}_{c_i} \cup
        \bld{B})$. Note that $\bld{M}_{c_i} \cup \bld{B} \subseteq
        \db(\dbiJ_i)$ and $V' \in \mcl{U}$ since $length(V') \leq m$. So,
        since by induction $b_i$
        and $c_i$ belong to the same equivalence class in $\dbiJ_i$, there
        exist tuple-sets
        $\bld{P}_{b_i}$ and $\bld{B}'$ with $\bld{P}_{b_i} \cup \bld{B}'
        \subseteq \dbiJ_i$ such that
        $b_i \in V'(\bld{P}_{b_i} \cup \bld{B}')$. [$\bld{P}_{b_i}$ and
        $\bld{B}'$, respectively, replace the
        role of $\bld{M}_{c_i}$ and $\bld{B}$.] Observe now that $a \notin
        adom(\bld{B})$ as $a = b_i$. Since $adom(\bld{M}_{b_i})
        \cap adom(\bld{A}) \subseteq
        \{b_i\}$ and $adom(\bld{A}) \cap adom(\bld{B}) = \emptyset$ from
        proposition \ref{prop:part}, it is easy to verify that
        \[
        a \in V(\bld{P}_{b_i} \cup \bld{A} \cup \bld{B}' \cup \bld{L}').
        \]
\item[case 2] $a \neq b_i$. This is divided into two further cases:
    \begin{description}
      \item[(a)] $b_i \in adom(A)$. This is divided into two further cases:
        \begin{description}
     \item[(i)] $a \in adom(A)$. In this case, note that
            $a \notin adom(\bld{M}_{b_i})$ (using Proposition \ref{prop:part}(5))
          and $a \notin adom(\bld{B}$) (using Proposition \ref{prop:part}(3)).
            We can then continue in the same fashion as in the case 1.
     \item[(ii)] $a \notin adom(A)$. Let $F$ be the set of all free tuples
        in the body of $V$ such that $\{ \nu(u) : u \in F \} =
        \bld{A}$. Let $\{y_1,\ldots,y_r\} \subseteq X$ be the set of
        variables in $F$ such that $\nu(y_j) = b_i$. Let $y$ be a new
        variable (i.e. $y \notin X$) and $F' := F[y/y_1,\ldots, y_r]$,
        i.e., we replace each occurrence of the variables $y_1, \ldots, y_r$
        in $F$ by $y$. Then, let $V'(y)$ be the view whose conjuncts are
        exactly $F'$:
        \[
        V'(y) \larw \bigwedge F'.
        \]
        Then, $V' \in \mcl{U}$ and $b_i \in V'(\bld{A})$. Since
        $\bld{A} \subseteq \db(\dbiJ_i)$ and because $b_i$ and $c_i$ belong to the
        same equivalence class in $\dbiJ_i$ (by the induction hypothesis),
        there exists a set $\bld{A}' \subseteq \db(\dbiJ_i)$ such that
        $c_i \in V'(\bld{A}')$.  Since $adom(\bld{M}_{b_i}) \cap
        adom(\bld{A}) \subseteq \{b_i\}$ and
        $adom(\bld{A}) \cap adom(\bld{B}) = \emptyset$ from proposition
        \ref{prop:part}, it is easy to check that
        $a \in V'(\bld{M}_{c_i} \cup \bld{A}' \cup \bld{B} \cup \bld{L}')$.
     \end{description}
       \item[(b)] $b_i \notin adom(A)$. Let $\bld{M}_{c_i} \assign
        \bld{M}_{b_i}[c_i/b_i]$. By construction, we see that
        $\bld{M}_{c_i} \subseteq \bld{S}_{c_i} \subseteq \db(\dbiJ_i)$. By
        proposition \ref{prop:part} (items 4 and 5), it is the case that
        \[
        a \in V(\bld{M}_{c_i} \cup \bld{A} \cup \bld{B} \cup \bld{L}').
        \]
    \end{description}
\end{description}
In any case, we have $a \in V(\dbiJ_i)$. This completes the proof.

\end{proof}
It remains to prove proposition \ref{prop:part}.
\begin{proof}[of proposition \ref{prop:part}]
The present situation is depicted in figure \ref{fig:part}.
\begin{figure}[h]
\begin{center}
    \epsfig{file=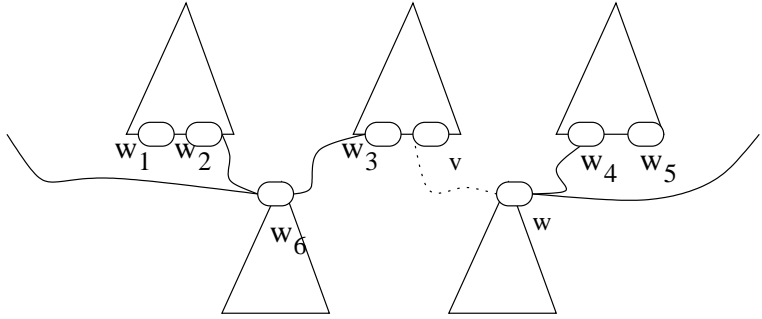}
\end{center}
\caption{\small $v$ is the $\bld{S}_{b_i}$-labeled node whose contents
    are to be changed by $\bld{S}_{c_i}[b_i/c_i]$. The node $w$ is
    $\bld{S}_{c_i}$-labeled, and will be ``linked'' to node $v$ after
    step $i+1$ of \texttt{prune} procedure is finished --- signified by
    the dotted line. Solid lines represent links that have been established
    in step $j < i+1$ of the procedure.
    \label{fig:part}}
\end{figure}
This is a snapshot of the moment just before we apply step $i+1$.
Step $i$ of \texttt{prune} procedure simply prunes the subtree rooted at
the $\bld{S}_{b_i}$-labeled node $v$, and links (rejustifies) $v$ using the
$\bld{S}_{c_i}$-labeled node $w$, where $b_i$ and $c_i$ belong to the same
equivalence class in $\bld{J}_i$. It is important to note that some
cousin\footnote{node of the same tree and level} $w_3$ of $v$ might also be
linked to a root $w_6$ of another tree, which in turn might be linked to
a leaf node $w_2$ of another tree, which in turn might have a cousin $w_1$ that
satisfies the same property as $w_3$ and so on. Furthermore, the node $w$ might
have also been linked to some other leaf $w_4$ that has a cousin $w_5$ that
is connected to a root of some other tree, and so on. Note that it is
impossible for two leafs of a tree to be linked to the same root node
of a tree by construction. Hence, the three trees in the middle (i.e. where
$v,w$, and $w_6$ are located) are necessarily distinct. The leftmost and
rightmost tree might be the same tree depending on the value of the girth $g$
that we defined earlier.

Let us now define
\begin{eqnarray*}
    \bld{A} & \assign & \{ t \in \bld{L} : d_{\gaif(\bld{J}_i)}(t,b_i)
                    \leq m \} \\
    \bld{B} & \assign & \{ t \in \bld{L} : d_{\gaif(\bld{J}_i)}(t,\bld{S}_{c_i})
                    \leq m \}.
\end{eqnarray*}
Intuitively, the set $\bld{A}$ contains tuples up to distance $m$ from the label
$\bld{S}_{b_i}$ of node $v$ in $\dbiJ_i$, while the set $\bld{B}$ contains
tuples up to distance $m$ from the label $\bld{S}_{c_i}$ of $w$ in $\dbiJ_i$. Note
that this is distance in the structure $\dbiJ_i$, not $\dbiJ_{i+1}$. It is
immediate that we have property (2) $\bld{A} \cup \bld{B} = \bld{L}$, as
the length of the view $V$ is at most $m$ and that
$adom(\bld{M}_{b_i}) \subseteq \{b_i\} \cup adom(\bld{S}_{c_i})$. So, it is
sufficient
to show that properties 3 and 5 are satisfied, as they obviously imply
properties 1 and 4. Note that our construction has ensured that:
\begin{enumerate}
\item Two nodes in any given tree in $\mcl{K}_i$ that are at least distance
    two apart cannot share a constant.
\item Two trees $T$ and $T'$ in $\mcl{K}_i$ cannot share a constant except on:
    (i) a unique leaf of $T$ and the root of $T'$, as is the case for
    $v$ and $w$ in Figure \ref{fig:part}
     or alternatively (ii) a unique leaf of
    $T$ and a unique leaf of $T'$.  This case can happen when both leafs are connected to the root
    of a different tree $T''$, as is the situation for $w_2$ and $w_3$ in Figure \ref{fig:part}.
\end{enumerate}

Therefore, for some sufficiently large constant $c' \in \mathbb{N}$, two nodes
$v'$ and $v''$ in $\mcl{K}_i$ of distance $c' m$ cannot have two elements of
$\bld{J}_i$ that are of distance $\leq m$ in $\gaif(\bld{J}_i)$. [In fact, a
careful analysis will show that $c' = 1$ is sufficient.] Therefore, the
locations of the constants in
$\bld{A}$ (resp. $\bld{B}$) cannot be ``very far away'' from the tuple $v$
(resp. $w$). In fact, if we set $c \geq c'$ (recall that $g = cm$) and consider the path $P$ between a
tuple $t \in \bld{A}$ and the constant $b_i$ (which belongs to $v$ and its
parent), it cannot
connect a root and a leaf of the same tree (i.e. through the body of the tree).
So, either it is completely contained in the tree of which $v$ is a leaf, or it
has to alternate alternate between leafs and root several times, and then end
in some tree. In figure \ref{fig:part}, we may pick the following example
\[
    v \rarw^* w_3 \rarw w_6 \rarw w_2 \rarw^* w_1 \rarw \ldots,
\]
where we use the notation $\rarw^*$ to mean ``path in the same tree''.
The same analysis can be applied to determine the locations of the tuples
of $\bld{B}$. Therefore, in order the ensure that properties 3 and 5 are
satisfied, we just need to ensure that the height of each tree and the girth of
$\mcl{K}_i$ be large enough, which can be done by taking a sufficiently
large $c$. When the girth (as ensured in the \texttt{copy} and \texttt{prune}
procedures) is sufficiently large, we can be sure that no paths of
length $\leq m$ exist between $v$ and $w$ in $\mcl{K}_i$
[In fact, a careful but tedious analysis shows that $c = 1$ is
sufficient.]
\end{proof}

Theorem \ref{th:fin_s} also holds for infinite models, since
even if the initial justification hierarchies are infinite, the proof method used
is unchanged.
We thus also
obtain finite controllability (every satisfiable formula is
finitely satisfiable) for $\ucv$.

\begin{proposition}
The $\ucv$  class of formulas
is finitely
controllable. \
\end{proposition}

\section{Extending the View Definitions}

\label{extend}

The previous section showed that the first order language using unary
conjunctive view definitions is decidable.  A natural way to increase
the power of the language is
to make view bodies more expressive (but retain
unary arity for the views).  We say earlier
that allowing unary views to use disjunction in their definition does not
actually increase expressiveness
of the UCV language and hence this case is decidable.
Unfortunately, as we will show, employing other ways
of extending the views results in satisfiability becoming undecidable.

The first extension we consider is allowing inequality in the
views,
e.g.,
\begin{grp}
$V(x) \leftarrow R(x,y), S(x,x), x \neq y$
\end{grp}
Call the first order language over such views the
{\em first order unary conjunctive$^{\neq}$ view language}.
In fact, this language allows 
us to check whether a two counter machine computation is valid
and terminates, which thus leads to the following result:

\begin{theorem}
\label{viewineq}
Satisfiability is undecidable for the first order unary conjunctive$^{\neq}$
view query language.
\end{theorem}
\begin{proof}
The proof is by a reduction from the halting problem of two counter machines
(2CM's) starting with zero in the counters.  Given any description of a 2CM
and its computation, we can show how to a) encode this description
in database relations and b) define queries to check this description.  We
construct a query which is satisfiable iff the 2CM
halts.  The basic idea of the simulation is similar to
one in \cite{LMSS92}, but with the major difference that  {\em
cycles are allowed} in the successor relation, though there must be at least
one good chain.

A two-counter machine is
a deterministic finite state machine with two non-negative counters.  The
machine can test whether a particular counter is empty or non-empty.
The transition function has the form
$$\delta: \; S \; \times \; \{=,>\} \; \times \; \{=,>\} \; \rightarrow \;
S \; \times \; \{pop,push\} \; \times \; \{pop,push\}$$
For example,  the statement $\delta(4,=,>) \; = \; (2,push,pop)$ means
that if we
are in state 4 with counter 1 equal to 0 and counter 2 greater
than 0, then go to state 2 and add one to counter 1 and subtract one
from counter 2.

The computation of the machine is stored in the relation
$config(t,s,c_1,c_2)$, where $t$ is the time, $s$ is the state and
$c_1$ and $c_2$ are values of the counters.
The states of the machine can be described by integers $0,1 \ldots
,h$ where 0 is the initial state and $h$ the halting (accepting) state.
The first configuration of the machine is $config(0,0,0,0)$ and
thereafter, for each move, the time is increased by one and the
state and counter values changed in correspondence with the transition
function.

We will use some relations to encode the computation of
2CMs starting with zero in the counters. These are:
\begin{itemize}
\item $S_0,\ldots,S_h$: each contains a constant which represents
      that particular state.
\item $succ$: the successor relation.  We will make sure it contains
one chain starting from $zero$
and ending at $last$ (but it may in addition contain unrelated cycles).
\item $config$: contains computation of the 2CM.
\item $zero$: contains the first constant in the chain in $succ$. This constant
is also used as the number zero.
\item $last$: contains the last constant in the chain in $succ$.
\end{itemize}

Note that we sometimes blur the distinction between unary relations
and unary views, since a view $V$ can simulate a unary relation $U$
if it is defined by $V(x) \leftarrow U(x)$.\\

The  unary and nullary views
(the latter can be eliminated using quantified unary views) are:
\begin{itemize}
\item $halt$: true if the machine halts.
\item $bad$:  true if the database doesn't correctly describe
      the computation of the 2CM.
\item $dsucc$:  contains all constants in $succ$.
\item $dT$:  contains all time stamps in $config$.
\item $dP$: contains all constants in $succ$ with predecessors.
\item $dCol_1, dCol_2$: are projections of the first and second columns
of $succ$.
\end{itemize}

When defining the views, we also state some formulas
(such as $hasPred$)
over the views which will be used to form our
first order sentence over the views.

\begin{itemize}
\item The ``domain'' views (those starting with the letter $d$) are easy to define, e.g.
 \begin{grp} $dP(x) \leftarrow succ(z,x)$\\
             $dCol_1(x) \leftarrow succ(x,y)$\\
             $dCol_2(x) \leftarrow succ(y,x)$
\end{grp}

\item $hasPred$ says ``each nonzero constant in $succ$
has a predecessor:''
\begin{grp}
$hasPred:$ $\forall x (dsucc(x) \Rightarrow (zero(x) \vee \ dP(x)))$
\end{grp}

\item $sameDom$ says
``the constants used in $succ$ and the timestamps in
$config$ are the same set'':
\begin{grp}
$sameDom:$ $\forall x (dsucc(x) \Rightarrow  dT(x)) \wedge
\forall y (dT(y) \Rightarrow dsucc(y)))$
\end{grp}

\item $goodzero$ says ``the zero occurs in $succ$'':
\begin{grp}
$goodzero:$ $\forall x (zero(x) \Rightarrow dsucc(x))$
\end{grp}

\item $nempty$ : each of the domains and
unary base relations is not empty
\begin{grp}
$nempty:$ $\exists x (dsucc(x))$
\end{grp}

\item Check that each constant in $succ$ has at most one successor
      and at most one predecessor and that it has no cycles of length 1.
\begin{grp}
$bad \leftarrow succ(x,y), succ(x,z), y \neq z$\\
$bad \leftarrow succ(y,x), succ(z,x), y \neq z$\\
$bad \leftarrow succ(x,x)$
\end{grp}

Note that the first two of these rules could be enforced by database style
functional
dependencies $x \rightarrow y$ and $y \rightarrow x$ on $succ$.

\item Check that every constant in the chain in succ which isn't
the last one must have a successor

\begin{grp}
$hassuccnext:$ $\forall y (dCol_2(y) \Rightarrow (last(y) \vee dCol_1(y))$
\end{grp}

\item Check that the last constant has no successor and
zero (the first constant)
has no predecessor.
\begin{grp}
$bad \leftarrow last(x), succ(x,y)$\\
$bad \leftarrow zero(x), succ(y,x)$
\end{grp}

\item Check that every constant eligible to be in last and zero must
be so.
\begin{grp}
$eligiblezero:$ $\forall y (dCol_1(y) \Rightarrow (dCol_2(y) \vee zero(y))$\\
$eligiblelast:$ $\forall y(dCol_2(y) \Rightarrow (dCol_1(y) \vee last(y)))$
\end{grp}

\item Each $S_i$ and $zero$ and $last$ contain $\leq$ 1 element.
\begin{grp}
$bad \leftarrow S_i(x), S_i(y), x \neq y$\\
$bad \leftarrow zero(x), zero(y), x \neq y$\\
$bad \leftarrow last(x), last(y), x \neq y$
\end{grp}

\item Check that $S_i,S_j
,last, zero$
are disjoint ($0 \leq i < j \leq h$):
\begin{grp}
$bad \leftarrow zero(x), last(x)$\\
$bad \leftarrow S_i(x), S_j(x)$\\
$bad \leftarrow zero(x), S_i(x)$\\
$bad \leftarrow last(x), S_i(x)$
\end{grp}

\item Check that the timestamp is the key for $config$.
There are three rules, one for the state and two for the two counters;
the one for the state is:
\begin{grp}
$bad \leftarrow config(t,s,c_1,c_2),$
$config(t,s',c'_1,c'_2), s \neq s'$
\end{grp}

\item Check the configuration of the 2CM at time zero.  $config$
must have a tuple at $(0,0,0,0)$ and there must not be any tuples
in config with a zero state and non zero times or counters.
\begin{grp}
$V_{z_s}(s)\!\leftarrow\!zero(t),\!config(t,s,x,y)$\\
$V_{z_{c_1}}(c)\!\leftarrow\!zero(t),\!config(t,x,c,y)$\\
$V_{z_{c_2}}(c)\!\leftarrow\!zero(t),\!config(t,x,y,c)$\\
$V_{y_s}(t)\!\leftarrow\!zero(s),\!config(t,s,x,y)$\\
$V_{y_{c_1}}(c_1)\!\leftarrow\!zero(s),\!config(t,s,c_1,x)$\\
$V_{y_{c_2}}(c_2)\!\leftarrow\!zero(s),\!config(t,s,x,c_2)$\\
$goodconfigzero:$ $\forall x (V_{z_s}(x) \Rightarrow S_0(x) \wedge $\\
\verb?   ?$(V_{z_{c_1}}(x) \vee V_{z_{c_2}}(x)
\vee V_{y_s}(x) \vee V_{y_{c_1}}(x)
\vee V_{y_{c_2}}(x)) \Rightarrow zero(x))$
\end{grp}

\item For each tuple in $config$ at time $t$ which isn't the halt state,
there must also be a tuple at time $t+1$ in $config$.
\begin{grp}
$V_1(t) \leftarrow config(t,s,c_1,c_2), S_h(s)$\\
$V_2(t) \leftarrow succ(t,t2), config(t2,s',c_1',c_2')$\\
$hasconfignext:$ $\forall t ((dt(t) \wedge \neg V_1(t)) \Rightarrow V_2(t))$
\end{grp}

\item Check that the transitions of the 2CM are followed.
For each transition $\delta(j,>,=) = (k,pop,push)$, we include three
rules, one for checking the state, one for checking the first counter
and one for checking the second counter.  For the transition in question we have
for checking the state
\begin{grp}
$V_{\delta}(t') \leftarrow  config(t,s,c_1,c_2),
succ(t,t'), S_j(s),succ(x,c_1),
zero(c_2)$\\
$V_{\delta_s}(s) \leftarrow V_{\delta}(t), config(t,s,c_1,c_2)$\\
$goodstate_{\delta}:$ $\forall s (V_{\delta_s}(s) \Leftrightarrow S_k(s))$
\end{grp}
and for the first counter, we (i)
find all the times where the transition is definitely correct
for the first counter
\[
  \begin{array}{@{}r@{~}l}
Q_{1_\delta}(t') & \leftarrow  config(t,s,c_1,c_2),\\
& \; \; \; succ(t,t'), S_j(s),succ(x,c_1),\\
& \; \; \;  zero(c_2),succ(c_1'',c_1),config(t',s',c_1'',c_2') \mbox{\Large}
  \end{array}
\]
(ii) find all the times where the transition may or may not be correct
for the first counter
\[
Q_{2_\delta}(t') \leftarrow  config(t,s,c_1,c_2),
succ(t,t'), S_j(s),succ(x,c_1),
zero(c_2)
\]
and make sure $Q_{1_\delta}$ and $Q_{2_\delta}$ are the same
\begin{grp}
$goodtrans_{\delta_{c_1}}:$ $\forall t (Q_{1_{\delta}}(t) \Leftrightarrow Q_{2_{\delta}}(t))$
\end{grp}
Rules for second counter are similar.

For transitions $\delta_1, \delta_2, \ldots , \delta_k$, the combination
can be expressed thus:
\begin{grp}
$goodstate:$ $goodstate_{\delta_1} \wedge goodstate_{\delta_2}
\wedge \ldots \wedge goodstate_{\delta_k}$\\
$goodtrans_{c_1}:$ $goodtrans_{\delta_{1_{c_1}}} \wedge goodtrans_{\delta_{2_{c_1}}}
\wedge \ldots \wedge goodtrans_{\delta_{k_{c_1}}}$\\
$goodtrans_{c_2}:$ $goodtrans_{\delta_{1_{c_2}}} \wedge goodtrans_{\delta_{2_{c_2}}}
\wedge \ldots \wedge goodtrans_{\delta_{k_{c_2}}}$
\end{grp}

\item Check that halting state is in $config$.
\begin{grp}
$hlt(t) \leftarrow config(t,s,c_1,c_2), S_h(s)$\\
$halt:$ $\exists x hlt(x)$
\end{grp}
\end{itemize}

Given these views, we claim that satisfiability is undecidable for
the query
$\psi = \neg bad \wedge hasPred \wedge sameDom \wedge halt
\wedge goodzero \wedge goodconfigzero \wedge \wedge nempty \wedge hassuccnext
\wedge eligiblezero \wedge eligiblelast \wedge
 goodstate \wedge goodtrans_{c_1} \wedge goodtrans_{c_2} \wedge hasconfignext$
\end{proof}

The second extension we consider is to allow ``safe'' negation in the
conjunctive views,
e.g.
\begin{grp}
 $V(x) \leftarrow R(x,y), R(y,z), \neg R(x,z)$
\end{grp}
Call the first order language over such views the
{\em first order unary conjunctive$^{\neg}$ view language}.
It is also undecidable, by a result in \cite{BDR98}.

\begin{theorem} \cite{BDR98}
Satisfiability is undecidable for the first order unary conjunctive$^{\neg}$
view query language. \ebox
\end{theorem}

A third possibility for increasing the expressiveness of views
would be to keep the body as a pure conjunctive query, but allow
views to have {\em binary} arity,  {\em e.g.}
\begin{grp}
$V(x,y) \leftarrow R(x,y)$
\end{grp}
This doesn't yield a decidable language either,
since this language has the same expressiveness as first
order logic over binary
relations, which is known to be undecidable \cite{BGG97}.

\begin{proposition}
Satisfiability is undecidable for the first order
binary conjunctive view language. \ebox
\end{proposition}

A fourth possibility is to use unary conjunctive views, but allow
recursive view definitions.  e.g.
\begin{grp}
$V(x) \leftarrow edge(x,y)$\\
$V(x) \leftarrow V(x) \wedge edge(y,x)$
\end{grp}
Call this the first order unary
conjunctive$^{rec}$ language.
This language is undecidable also.

\begin{theorem}
Satisfiability is undecidable for the first order
unary conjunctive$^{rec}$ view language.
\end{theorem}
\begin{proof} (sketch):  The proof of theorem \ref{viewineq} can be adapted
by removing inequality and instead using recursion
to ensure there exists a connected chain in $succ$.  It then becomes more complicated, but
the main property needed is that $zero$ is connected
to $last$ via the constants in $succ$.  This can be expressed by
\begin{grp}
$conn\_zero(x) \leftarrow zero(x)$\\
$conn\_zero(x) \leftarrow conn\_zero(y), succ(y,x)$\\
$\exists x (last(x) \wedge conn\_zero(x))$
\end{grp}
\end{proof}

\section{Applications}
\label{sec:app}

\subsection{Reasoning Over Ontologies}

A currently active area of research
is that of reasoning over ontologies (see e.g. \cite{DBLP:conf/iccs/Horrocks05}).
The aim here is to use decidable query languages
used for accessing and reasoning about information and structure for the Semantic Web.
In particular, ontologies provide vocabularies which can define
relationships or associations between various concepts (classes)
and also properties that link different classes together.
Description logics are a key
tool for reasoning over schemas and ontologies and to this end, a considerable
number of different description logics have been developed.
To illustrate some reasoning over a simple ontology, we
adopt an example from \cite{horr03}, describing people, countries and some relationships.
This example can be encoded
in a description logic such as $SHIQ$ and also in the UCV query language.
We show how to accomplish the latter.

\begin{itemize}
\item Define classes such as $Country, Person, Student$ and
$Canadian$.  These are just unary views defined over unary relations, e.g.
$Country(x) \leftarrow country(x)$.   Observe that we
can blur the distinction between unary views and unary relations and use them interchangeably.
\item State that $student$ is a subclass of $Person$.
$$ \forall x Student(x) \Rightarrow Person(x)$$
\item  State that $Canada$ and $England$ are both instances of the class
$Country$.  To accomplish this in the UCV language, we could define $Canada$ and
$England$ as unary views and ensure that they are contained in the $Country$
relation and are disjoint with all other classes/instances.
\item Declare $Nationality$ as a property relating the classes $Person$
(its domain) and $Country$ (its range).  In the UCV language, we could model this as
a binary relation $Nationality(x,y)$ and impose constraints on its domain and
range. e.g.
\begin{grp}
$dom\_Nationality(x) \leftarrow Nationality(x,y)$\\
$range\_Nationality(y) \leftarrow Nationality(x,y)$\\
$\forall x (dom\_Nationality(x) \Rightarrow Person(x))$\\
$\forall x (range\_Nationality(x) \Rightarrow Country(x))$\\
\end{grp}
\item State that $Country$ and $Person$ are disjoint classes.
$\forall x (Country(x) \Rightarrow \neg Person(x))$.
\item Assert that the class $Stateless$ is defined precisely
as those members of the class $Person$ that have no values for the
property $Nationality$.
\begin{grp}
$has\_Nationality(x) \leftarrow Nationality(x,y)$\\
$Stateless (x) \Leftrightarrow Person(x) \wedge \neg has\_Nationality(x)$\\
\end{grp}

\end{itemize}

The above types of statements are reasonably simple to express.
In order to achieve more expressiveness,
property chaining and property composition have been identified as important reasoning features.
To this end, integration of rule-based KR and DL-based KR is an active area of research.
The UCV query language has the advantage of being able to express certain types
of property chaining, which
would not be expressible in the description logic SHIQ, which is not able to
accomplish chaining \cite{horr03}.  For example

\begin{itemize}
\item An uncle is precisely a parent's brother.
\begin{grp}
$uncle_1(z) \leftarrow parent(x,y), brother(x,z)$\\
$uncle_2(z) \leftarrow parent(x,y), brother(z,x)$\\
$uncle(z) \Leftrightarrow uncle_1(z) \vee uncle_2(z)$\\
\end{grp}
\end{itemize}

We consequently
believe the UCV query language has some intriguing potential to be used as a reasoning component
for ontologies, possibly to supplement description logics for some specialized applications.
We leave this as
an open area for future investigation.


\subsection{Containment and Equivalence}
We now briefly examine the application of our results to
query containment.
Theorem \ref{th:fin_s} implies we can test whether $Q_1(x) \subseteq Q_2(x)$ under the constraints
$C_1 \wedge C_2 \ldots \wedge C_n$ where $Q_1, Q_2, C_1, \ldots, C_n$ are
all first order unary conjunctive view queries in 2-NEXPTIME.
This just amounts
to testing whether the sentence $\exists x (Q_1(x) \wedge \neg Q_2(x))  \wedge
C_1 \wedge \ldots \wedge C_n$ is unsatisfiable.  Equivalence of $Q_1(x)$
and $Q_2(x)$ can be tested with containment tests in both directions.

Of course, we can also show that testing the containment
$Q_1 \subseteq Q_2$ is undecidable if  $Q_1$ and $Q_2$ are first order
unary conjunctive view$^{\neq}$ queries, first order unary
conjunctive view$^{\neg}$ queries and first order unary conjunctive$^{rec}$
view queries.  

Containment of queries with negation was first considered in \cite{SY80}.  There it was essentially shown that the
problem is decidable for queries which do not apply projection
to subexpressions with difference.  Such a language is disjoint from ours,
since it cannot express a sentence such as $ \exists y V_4(y) \wedge
\neg\exists x (V_1(x) \wedge
\neg V_2(x))$ where $V_1$ and $V_2$ are views defined over several variables.

\subsection{Inclusion Dependencies}

Unary inclusion dependencies were identified as useful in \cite{CKV90}.
They take the form $R[x] \subseteq S[y]$.
If we allow $R$ and $S$ above to be unary conjunctive view queries,
we could obtain {\em unary conjunctive view containment dependencies}.
Observe that the unary views are actually
unary projections of the join of one or more relations.

We can also define a special type
of dependency called a {\em proper}  first order unary conjunctive inclusion dependency,
having the form
$Q_1(x) \subset Q_2(x)$, where $Q_1$ and $Q_2$ are first order unary conjunctive
view queries with one free variable.  
If $\{d_1, \ldots, d_k\}$ is a set of such dependencies, then it is
straightforward to test whether they imply another
dependency $d_x$, by testing the satisfiability of an appropriate
first order unary conjunctive view query.

\begin{theorem}
\label{implication}
Implication 
for the class of
unary conjunctive view containment dependencies
with subset and proper subset operators 
is i) decidable in 2-NEXPTIME and ii) finitely controllable. \ebox
\end{theorem}

The results from \cite{CKV90} show
that implication is decidable in polynomial time, but
not finitely controllable,
for either of the combinations
i) functional dependencies plus unary inclusion dependencies, ii) 
full implication dependencies plus unary inclusion dependencies.  In contrast,
the stated 
complexity in the above theorem is much higher, due to the increased
expressiveness of the dependencies, yet interestingly the class is finitely
controllable.   
 
We might also consider
unary conjunctive$^{\neq}$ containment
dependencies.  The tests in the proof of theorem \ref{viewineq}
for the 2CM can be written in the form $Q_1(x) \subseteq Q_2(x)$, with the
exception of the non-emptiness constraints, which must use the proper
subset operator.
Interestingly also, we can see from
the proof of theorem \ref{viewineq}, that adding the ability to express
functional dependencies 
would also result in undecidability.
We can summarise these observations in the following theorem and its corollary.

\begin{theorem}
Implication is undecidable for unary conjunctive$^{\neq}$
(or conjunctive$^{\neg}$) view containment dependencies
with the subset and the proper subset operators. \ebox
\end{theorem}

\begin{corollary}
Implication is undecidable for the combination of
unary conjunctive view containment 
dependencies
plus functional dependencies.
\end{corollary}

\subsection{ Active Rule Termination}


The languages in this paper have their origins in \cite{BDR98}, where
active database rule languages based on views were studied.
The decidability result for first
order unary conjunctive views can be used to positively answer an open
question raised in \cite{BDR98}, which essentially asked whether termination
is decidable for active database rules expressed using unary conjunctive views.


\section{Expressive Power of the UCV Language}
\label{sec:expressive}
As we have seen in the previous sections, the logic $\ucv$ is quite suitable
to reason about hereditary information such as ``$x$ is a grandchild of $y$''
over family trees. This is due to the fact that $\ucv$ can express the
existence of a directed walk of length $k$ in the graph, for any fixed positive
integer $k$. Therefore, it is natural to also ask what is inexpressible in the logic.
In this section, we describe a game-theoretic technique for proving
inexpressibility results for $\ucv$. First, we show an easy adaptation of
\ef $\;$ games for proving that a boolean
query is inexpressible in $\ucv(\sigma,\mcl{V})$ for a signature $\sigma$ and a
finite view set $\mcl{V}$ over $\sigma$. Second, we extend this result for
proving that a boolean query is inexpressible in $\ucv(\sigma)$. An
inexpressibility result of the second kind is clearly more interesting, as it
is independent of our choice of the view set $\mcl{V}$ over $\sigma$. Moreover, such a result places an ultimate limit of what can be expressed by UCV queries.
Although it can be adapted to any class $\mcl{C}$ of
structures, we shall only state our theorem for proving inexpressibility
results in $\ucv$ \emph{over all finite structures}. For this section only, we
shall use $STRUCT(\sigma)$ to denote the set of all \emph{finite}
$\sigma$-structures.

Our first goal is quite easy to achieve. Recall that each view set $\mcl{V}$
over $\sigma$ induces a mapping $\Lambda: STRUCT(\sigma) \rarw STRUCT(\mcl{V})$
as defined in section \ref{sec:prelim}.
\begin{theorem}
Let $\bld{A}, \bld{B} \in STRUCT(\sigma)$. Define the function $\Lambda:
STRUCT(\sigma) \rarw STRUCT(\mcl{V})$. Then, the following statements
are equivalent:
\begin{enumerate}
\item $\bld{A}$ and $\bld{B}$ agree on $\ucv(\sigma,\mcl{V})$.
\item $\Lambda(\bld{A}) \equiv_1^{\ufo(\mcl{V})} \Lambda(\bld{B})$ (i.e. they
    agree on $\ufo(\mcl{V})$ formulas of quantifier rank 1.
\end{enumerate}
\end{theorem}
\begin{proof}
Immediate from lemma \ref{lm:unaryform}, and lemma \ref{lem:realize}.
\end{proof}
So, to prove that a boolean query $\mcl{Q}$ is not expressible in
$\ucv(\sigma,\mcl{V})$, it suffices to find two $\sigma$-structures
such that $\Lambda(\bld{A}) \equiv_1^{\ufo(\mcl{V})} \Lambda(\bld{B})$, but
$\bld{A}$ and $\bld{B}$ do not agree on $\mcl{Q}$. In turn, to show that
$\Lambda(\bld{A}) \equiv_1^{\ufo(\mcl{V})} \Lambda(\bld{B})$, we can use
Ehrenfeucht-Fra\"isse games.

We now turn to the second task. Let us begin by stating an obvious corollary of
the preceding theorem.
\begin{corollary}
Let $\bld{A}, \bld{B} \in STRUCT(\sigma)$. For any view set $\mcl{V}$, define
the function $\Lambda^{\mcl{V}}: STRUCT(\sigma) \rarw STRUCT(\mcl{V})$. Then,
the following statements are equivalent:
\begin{enumerate}
\item $\bld{A}$ and $\bld{B}$ agree on $\ucv(\sigma)$.
\item For any view set $\mcl{V}$ over $\sigma$, we have
$\Lambda^{\mcl{V}}(\bld{A}) \equiv_1^{\ufo(\mcl{V})} \Lambda^{\mcl{V}}(\bld{B})$
\end{enumerate}
\label{cor:inexpr1}
\end{corollary}
This corollary is not of immediate use. Namely, checking the second statement
is a daunting task, as there are infinitely
many possible view sets $\mcl{V}$ over $\sigma$. Instead, we shall propose a
sufficient condition for this, which employs the easy direction of the
well-known homomorphism preservation theorem (see \cite{Hodg97}).
\begin{definition}
    A formula $\phi$ over a vocabulary $\sigma$ is said to be
    \defn{preserved under homomorphisms}, if for any $\bld{A}, \bld{B} \in
    STRUCT(\sigma)$ the following statement holds: whenever $\bld{a} \assign
    (a_1,\ldots,a_m) \in
    \phi(\bld{A})$ and $h$ is a homomorphism from $\bld{A}$ to $\bld{B}$,
    it is the case that $h(\bld{a}) \assign (h(a_1),\ldots,h(a_m)) \in
    \phi(\bld{B})$.
\end{definition}
\begin{lemma}
Conjunctive queries are preserved under homomorphisms.
\label{lm:homom}
\end{lemma}

\begin{theorem}
    Let $\bld{A}, \bld{B} \in STRUCT(\sigma)$. To prove that
    $\Lambda(\bld{A}) \equiv_1^{\ufo(\mcl{V})} \Lambda(\bld{B})$ for all
    $\sigma$-view sets $\mcl{V}$, it is sufficient to show that
    \begin{enumerate}
    \item For every $a \in A$, there exists a homomorphism $h$ from
        $\bld{A}$ to $\bld{B}$ and a homomorphism $g$ from $\bld{B}$ to
        $\bld{A}$ such that $g(h(a)) = a$.
    \item For every $b \in B$, there exists a homomorphism $h$ from
        $\bld{A}$ to $\bld{B}$ and a homomorphism $g$ from $\bld{B}$ to
        $\bld{A}$ such that $h(g(b)) = b$.
    \end{enumerate}
\label{thm:homom}
\end{theorem}
\begin{proof}
Take an arbitrary $\sigma$-view set $\mcl{V}$. We use Ehrenfeucht-Fra\"isse game
argument. Suppose Spoiler places a pebble on an element $a$ of
$\Lambda(\bld{A})$, whose domain is $A$. Then, the first assumption tells us
that there exist homomorphisms $h: A \rarw B$ and $g: B \rarw A$ such that
$g(h(a)) = a$. Duplicator
may respond by placing the other pebble from the same pair on the element
$h(a)$ of $\Lambda(\bld{B})$. To show this, we need to prove that
$a \mapsto h(a)$
defines an isomorphism between the substructures of $\Lambda(\bld{A})$ and
$\Lambda(\bld{B})$ induced by, respectively, the sets $\{a\}$ and $\{h(a)\}$.
Let $V \in \mcl{V}$.
It is enough to show that $a \in V(\bld{A})$ iff $h(a) \in V(\bld{B})$.
If $a \in V(\bld{A})$, then we have $h(a) \in V(\bld{B})$ by lemma
\ref{lm:homom}. Similarly, if $h(a) \in V(\bld{B})$, theorem \ref{lm:homom}
implies that $a = g(h(a)) \in V(\bld{A})$.

For the case where Spoiler plays an element of $\bld{B}$, we can use the
same argument with the aid of the second assumption above. In either case,
we have $\Lambda(\bld{A}) \equiv_1 \Lambda(\bld{B})$.
\end{proof}
This theorem allows us to give easy inexpressibility proofs for a variety
of first-order queries. We now give three easy inexpressibility proofs for
first-order queries over directed graphs (i.e. structures with one binary
relation $E$).
\begin{example}
    We show that the formula $SYM \equiv \forall x,y(E(x,y) \lrarw E(y,x))$
    accepting graphs with symmetric $E$ is
    not expressible in $\ucv(\sigma)$. To do this, consider the graphs
    $\bld{A}$ and $\bld{B}$ defined as follows
    \begin{center}
    \epsfig{file=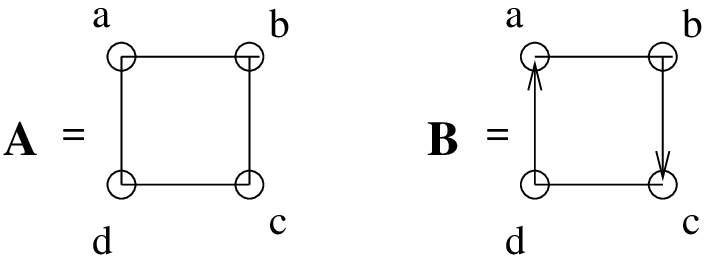}
    \end{center}
    Obviously, the graph $E^{\bld{A}}$ is symmetric, while $E^{\bld{B}}$
    is not. Consider the functions $h_1,h_2: A \rarw B$ and
    $g: B \rarw A$ defined as
    \begin{itemize}
    \item $h_1(a) = h_1(c) = a$ and $h_1(b) = h_1(d) = b$,
    \item $h_2(a) = h_2(c) = c$ and $h_2(b) = h_2(d) = d$, and
    \item for $i \in B$, $g(i) = i$.
    \end{itemize}
    It is easy to verify that $h_1$ and $h_2$ are homomorphisms from $\bld{A}$
    to $\bld{B}$, whereas $g$ a homomorphism from $\bld{B}$ to $\bld{A}$.
    Now, for $x \in \{a,b\}$, we have $g(h_1(x)) = x$ and $h_1(g(x)) = x$.
    For $x \in \{c,d\}$, we have $g(h_2(x)) = x$ and $h_2(g(x)) = x$. So,
    by theorem \ref{thm:homom} and corollary \ref{cor:inexpr1}, we conclude that
    $SYM$ is not expressible in $\ucv(\sigma)$ over all finite directed
    graphs.
\end{example}
\begin{example}
    We now show that the transitivity query
    \[
    TRANS \equiv \forall x,y,z( E(x,y) \wedge E(y,z) \rarw E(x,z) )
    \]
    is not expressible in $\ucv(\sigma)$. To do this, consider the graphs
    $\bld{A}$ and $\bld{B}$ defined as
    \begin{center}
    \epsfig{file=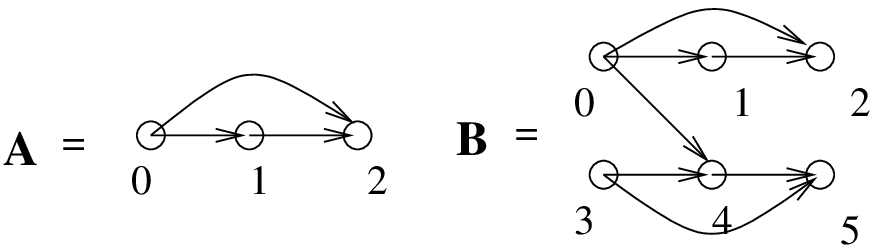}
    \end{center}
    It is obvious that $\bld{A} \models TRANS$, and it is not the case that
    $\bld{B} \models TRANS$. Consider the homomorphisms $h_1, h_2$ from
    $\bld{A}$ to $\bld{B}$, and the homomorphism $g$ from $\bld{B}$ to $\bld{A}$
    defined as
    \begin{itemize}
    \item for $i \in A$, $h_1(i) = i$;
    \item for $i \in A$, $h_2(i) = i+3$; and
    \item for $i \in B$, $g(i) = i \bmod{3}$.
    \end{itemize}
    Then, for $i \in A$, we have $g(h_1(i)) = i$. Conversely, suppose
    that $i \in B$. If $i = 0,1,2$, then $h_1(g(i)) = i$. Similarly, if
    $i = 3,4,5$, then $h_2(g(i)) = i$. So, by theorem \ref{thm:homom} and
    corollary \ref{cor:inexpr1}, transitivity is not expressible in
    $\ucv(\sigma)$ over finite directed graphs.
\end{example}
\begin{example}
The query $\forall x,y E(x,y)$ is also not expressible in $\ucv(\sigma)$.
It is easy to apply theorem \ref{thm:homom} and corollary \ref{cor:inexpr1}
on the following graphs to verify this fact.
\begin{center}
    \epsfig{file=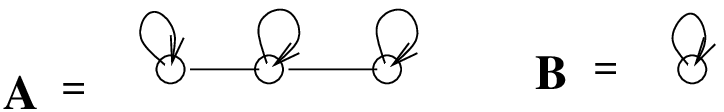}
\end{center}
\end{example}

\section{Related Work}
\label{sec:related}
Satisfiability of first order logic has been thoroughly
investigated in the context of the classical decision problem \cite{BGG97}.
The main thrust there has been determining for which quantifier
prefixes first order languages are decidable.  We are not aware
of any result of this type which could be used to demonstrate
decidability of the  first order unary conjunctive view language.
Instead, our result is best classified as a new
decidable class generalising the traditional
decidable unary first-order language (the L\"{o}wenheim class \cite{lowenheim}).
Use of the L\"{o}wenheim class itself for reasoning
about schemas is described in \cite{Th96}, where applications towards
checking intersection and disjointness of object oriented classes are given.

As observed earlier,
description
logics are important logics for
expressing constraints on desired models.
In \cite{CDL98},
the query containment problem is studied in the context of
the description logic ${\cal DLR}_{reg}$.  There are certain similarities
between this and the first order (unary) view languages  we have studied in
this paper.
The key difference appears to be that although ${\cal DLR}_{reg}$ can be
used to define view constraints, these constraints cannot express
unary conjunctive views (since assertions do not allow arbitrary projection).
  Furthermore, ${\cal DLR}_{reg}$ can express
functional dependencies on a single attribute, a feature which
would make the UCV language undecidable (see proof of theorem \ref{viewineq}).
There is a result in \cite{CDL98}, however, showing undecidability
for a fragment of ${\cal DLR}_{reg}$ with inequality, which could be adapted
to give an alternative proof of
theorem \ref{viewineq} (although inequality is used there in a slightly more
powerful way).

Another interesting family of decidable logics are
guarded logics.
The Guarded Fragment \cite{guarded98} and the Loosely Guarded
Fragment \cite{bentham97}
are both logics that have the finite model property \cite{DBLP:journals/sLogica/Hodkinson02}.
The philosophy of UCV is somewhat
similar
to these guarded logics, since the decidability of UCV also arises from certain
restrictions on quantifier use.
In terms of expressiveness though, guarded
logics
seem distinct from UCV formulas, not being able to express
cyclic views, such as
$\exists x (V(x))$, where $V(x) \leftarrow R(x,y), R(y,z), R(z,z'), R(z',x)$.

Another area of work that deals with complexity of views is the view consistency problem, with
results given in  \cite{AD98}.
This involves determining
whether there exists an underlying database instance that realises a
{\em specific} (bounded) view instance .  The problem we have focused on in this
paper is slightly more complicated; testing satisfiability of a first
order view query
asks the question whether there exists an ({\em unbounded}) view instance
that makes the query true.    This explains how satisfiability can be undecidable
for first order unary conjunctive$^{\neq}$ view queries, but view
consistency for non recursive datalog$^{\neq}$ views is in $NP$.
Monadic views have been recently examined in \cite{DBLP:conf/icdt/NashSV07}, where they
were shown to exhibit nice properties in the context of
answering and rewriting conjunctive queries using only a set of views.  This is an
interesting counterpoint to the result of this paper, which demonstrate how monadic
views can form the basis of a decidable fragment of first order logic.

\section{Summary and Further work}
\label{sec:summary}

In this paper, we have introduced a new decidable language based on the use of 
unary conjunctive
views embedded within first order logic.  This is a powerful generalisation of
the well known fragment of first order logic using only unary relations (the L\"{o}wenheim class).
We also showed that
our new class is maximal, in the sense that increasing the expressivity of
views
is not possible without undecidability resulting.
Table~1 provides a summary of our decidability results.  Note that the Unary
Conjunctive$^{\cup}$ View language corresponds to the extension of $\ucv$ by
allowing disjunction in the view definition.

We feel that the decidable case we have identified,
is sufficiently natural and interesting to be of practical, as well as theoretical interest.
\begin{table}
\begin{center}
{\large Table 1:\\
Summary of Decidability Results for First Order View Languages}\\
\begin{tabular}{||l|l||} \hline
Unary Conjunctive View  & Decidable\\ \hline
Unary Conjunctive$^{\cup}$ View & Decidable\\ \hline
Unary Conjunctive$^{\neq}$ View  & Undecidable\\ \hline
Unary Conjunctive$^{rec}$ View  & Undecidable\\ \hline
Unary Conjunctive$^{\neg}$ View  & Undecidable \cite{BDR98}\\ \hline
Binary Conjunctive View  & Undecidable \\ \hline
\end{tabular}
\label{summary-table}
\end{center}
\end{table}

An interesting open problem for future work is to investigate the decidability of
an extension to the first order unary conjunctive view language, when
equality is allowed to be used outside of the unary views (i.e. included in the first
order part).    An example formula in this
new language is\\

$$ \forall X,Y (V_1(X) \wedge V_2(Y) \Rightarrow X \neq Y) $$

We conjecture this extended language is decidable, but do not currently have a proof.

For other future work, we believe it would be worthwhile to investigate relationships
with description logics and also examine alternative ways of introducing
negation into the UCV language.  One possibility might be to allow views of arity
zero to specify description logic like constraints, such as $R_1(x,y)
\subseteq R_2(x,y)$.

Finally, there is still an exponential gap between the upper bound complexity of
2-NEXPTIME and lower bound complexity of NEXPTIME-hardness that we derived. The
primary reason for this exponential blow-up is the enumeration of all subviews
of the views that are present in the formula, which we need for the proof.

\begin{acks}
We thank Sanming Zhou for pointing out useful references on extremal graph
theory. We are grateful to Leonid Libkin for his comments on a draft of this paper.
\end{acks}

\bibliographystyle{acmtrans}
\bibliography{mvq}

% ...
\end{document}